\documentclass[conference]{IEEEtran}
\IEEEoverridecommandlockouts

\usepackage{cite}
\usepackage{amsmath,amssymb,amsfonts}
\usepackage{algorithmic}
\usepackage{graphicx}
\usepackage{latexsym}
\usepackage{xcolor}
\usepackage{comment}
\usepackage{cite}
\usepackage{url}
\usepackage[T3,T1]{fontenc}
\DeclareSymbolFont{tipa}{T3}{cmr}{m}{n}
\DeclareMathAccent{\invbreve}{\mathalpha}{tipa}{16}
\usepackage{newtxtext}
\usepackage[varg]{newtxmath}
\usepackage{moreverb}
\usepackage{upgreek}
\usepackage{mathrsfs}
\usepackage{mathtools}
\usepackage{bm}
\def\BibTeX{{\rm B\kern-.05em{\sc i\kern-.025em b}\kern-.08em
    T\kern-.1667em\lower.7ex\hbox{E}\kern-.125emX}}

\newtheorem{theorem}{Theorem}
\newtheorem{lemma}{Lemma}
\newtheorem{proposition}{Proposition}
\newtheorem{property}{Property}
\newtheorem{remark}{Remark}
\newtheorem{definition}{Definition}
\newtheorem{corollary}{Corollary}

\allowdisplaybreaks

\newcommand{\phin}{\phi^{(n)}}
\newcommand{\phinX}{\phi^{(n)}(\bm{X})}
\newcommand{\phinx}{\phi^{(n)}(\bm{x})}

\newcommand{\psin}{\psi^{(n)}}

\newcommand{\Phin}{\Phi^{(n)}}
\newcommand{\PhiKn}{\Phi_{\bm{K}}^{(n)}}
\newcommand{\PhiKnX}{\Phi_{\bm{K}}^{(n)}(\bm{X})}
\newcommand{\PhiXnK}{\Phi_{\bm{X}}^{(n)}(\bm{K})}
\newcommand{\Phiknx}{\Phi_{\bm{k}}^{(n)}(\bm{x})}

\newcommand{\Psin}{\Psi^{(n)}}
\newcommand{\PsinK}{\Psi_{\bm{K}}^{(n)}}

\newcommand{\Dn}{\mathcal{D}^{(n)}}

\newcommand{\MI}{\Delta_{\mathrm{MI}}^{(n)}}

\newcommand{\hugel}{{\arraycolsep 0mm
                    \left\{\ba{l}{\,}\\{\,}\ea\right.\!\!}}

\newcommand{\huger}{{\arraycolsep 0mm
                    \left.\ba{l}{\,}\\{\,}\ea\!\!\right\}}}

\newcommand{\hugebl}{{\arraycolsep 0mm
                    \left[\ba{l}{\,}\\{\,}\ea\right.\!\!}}

\newcommand{\hugebr}{{\arraycolsep 0mm
                    \left.\ba{l}{\,}\\{\,}\ea\!\!\right]}}

\newcommand{\defeq}{:=}

\newcommand{\beq}{\begin{equation}}
\newcommand{\eeq}{\end{equation}}
\newcommand{\beqa}{\begin{eqnarray}}
\newcommand{\eeqa}{\end{eqnarray}}
\newcommand{\beqno}{\begin{eqnarray*}}
\newcommand{\eeqno}{\end{eqnarray*}}
\newcommand{\ba}{\begin{array}}
\newcommand{\ea}{\end{array}}

\newcommand{\MEq}[1]{\stackrel{
{\rm (#1)}}{=}}

\newcommand{\MLeq}[1]{\stackrel{
{\rm (#1)}}{\leq}}

\newcommand{\MGeq}[1]{\stackrel{
{\rm (#1)}}{\geq}}

\newcommand{\MG}[1]{\stackrel{
{\rm (#1)}}{>}}

\newcommand{\MRar}[1]{\stackrel{
{\rm (#1)}}{\Rightarrow}}

\newcommand{\MIn}[1]{\stackrel{
{\rm (#1)}}{\in}}

\begin{document}

\title{

\newcommand{\ISITATitle}{ 
Source Encryption under Mutual Information Security 
Criterion: 
Universal Coding, Strong Converse Theorem
}
\newcommand{\ArXivTitle}{
}{
A Framework of Secure Source Coding 
using Mutual Information Security Criterion:  
Universal Coding, Strong Converse Theorem
}
\newcommand{\IEICETitle}{
A Framework of Source Encryption using Mutual Information
Security Criterion and the Strong Converse Theorem
}

}
\author{
	\IEEEauthorblockN{Yasutada Oohama and Bagus Santoso}
	\IEEEauthorblockA{University of Electro-Communications, Tokyo, Japan\\ 
	Email: \url{{oohama,santoso.bagus}@uec.ac.jp}}
%
}

\maketitle

\begin{abstract}
In this paper, we propose a framework of source encryption, 
where cryptographic processing 
is applied to a prescribed 
fixed length source code. 
The proposed source encryption framework is 
based on the secure communication framework 
of the Shannon cipher system. 
In the proposed framework, we use 
the mutual information as a measure 
of information leakage to an adversary.
For the proposed framework, we explicitly 
establish 
the necessary and 
sufficient 
condition for reliable and secure communication 
under the condition that error probability and
information leakage, respectively, 
are upper 
bounded by prescribed constants 
$\varepsilon\in (0,1)$ and $\delta \in (0,\infty)$.  
We also show that the obtained necessary 
and sufficient condition does not depend 
on the constants 
$\varepsilon\in (0,1)$ and $\delta\in (0,\infty)$, 
demonstrating that we have the strong converse 
theorem 
for the proposed framework of source encryption.
We further prove the existence
of encryption/decryption schemes, which are universal in the
sense that they work effectively for any distributions of the plain
text and those of the key used for the encryption.
\newcommand{\OmitZs}
{
In this paper we consider the variable-length 
lossless source coding for discrete 
memoryless sources. 
We proposes a new encryption framework 
for securely transmitting codewords 
over a noiseless channel. 
The proposed source encryption framework is 
based on the secure communication framework 
of the Shannon cipher system. 
In the proposed framework, we use the mutual information 
as a measure of information leakage to an adversary. 
We further prove the existence
of encryption/decryption schemes, which are universal in the
sense that they work effectively for any distributions of the plain
text and those of the key used for the encryption. 

}

\end{abstract}

\newcommand{\AsbstISITATwoSix}{
we focus on our attention to
strengthening the direct coding theorem. We prove the existence
of encryption/decryption schemes, which are universal in the
sense that they work effectively for any distributions of the plain
text, any noisy channels through which the adversary observe the
corrupted version of the key, and any measurement device used
for collecting the physical information. Those schemes have a
good performance such that if we compress the ciphertext with
rate within the reliable and secure rate region, then: (1) anyone
with secret key will be able to decrypt and decode the ciphertext
correctly, but (2) any adversary who obtains the ciphertext and
also the side physical information will not be able to obtain
any information about the hidden source as long as the leaked
physical information is encoded with a rate within the rate
constraint.
}


\section{Introduction}

Source coding is widely used as a technique 
of encoding information in order to enhance   
the efficiency of information transmission under 
high reliability on the transmission. 
In addition to such efficiency in 
information transmission, security in communication, namely, 
preventing information leakage over public 
channels, is also an important problem.
Shannon~\cite{Shannon} formulated cryptosystems 
within an information-theoretic framework and 
clarified the concept and fundamental limits 
of perfect secrecy for shared-key encryption 
schemes. Since then, a variety of studies extending
Shannon's arguments have been conducted. 
For example, Yamamoto~\cite{Yamamoto1} focused 
on correlated information sources and formulated 
the key rate required for encryption from 
the viewpoint of the common information of the sources, and 
further extended this framework to rate-distortion 
theory~\cite{Yamamoto2}. Furthermore, 
Hayashi and Yamamoto~\cite{Yamamoto3} introduced new security 
criteria based on the eavesdropper's number 
of guesses and probability of correct decoding.

In this paper we consider the fixed-length 
source coding for discrete memoryless sources. 
We proposes a new encryption framework for securely 
transmitting codewords 
over a noiseless channel. 
The proposed source encryption framework is based 
on the secure communication framework 
of the Shannon cipher system. 

As previous works, Oohama and Santoso 
\cite{DBLP:conf/isit/OohamaS22},
\cite{DBLP:conf/isit/OohamaS24}
investigated common key cryptsystem involving  
side channel attacks.
%
Oohama and Santoso~
\cite{DBLP:conf/itw/OohamaS21},
\cite{DBLP:conf/isita/OohamaS22}, and 
\cite{Oohama2025} 
investigated the distributed source encryption 
and obtained the necessary and sufficient 
conditions for achieving secure communication. 
Our study is motivated by those previous works. 

We propose a framework of ``source encryption,'' 
in which encryption is applied to a given 
source encoder and decoder. For the source 
encoder, we consider the case where 
fixed-length source sequences are mapped 
to fixed-length binary sequences 
(fixed-length source coding).
In the proposed framework, we use 
the mutual information as a measure of 
information leakage to an adversary.  
The maximum mutual information used in the 
previous works 
\cite{DBLP:conf/isit/OohamaS22},
\cite{DBLP:conf/isita/OohamaS22}
has the advantage that it provides a security 
measure determined solely by the 
encryption system, independent of the 
source characteristics.
On the other hand, when one aims to evaluate 
the security of information transmission 
or a specific source, the maximum mutual 
information cannot be used as a security 
evaluation criterion.

For the proposed framework we explicitly 
establish the necessary and sufficient 
condition for reliable and secure communication 
under the condition that error probability and
information leakage, respectively, are upper 
bounded by prescribed constants 
$\varepsilon\in (0,1)$ and $\delta \in (0,\infty)$.  
We also show that the obtained necessary 
and sufficient condition does not depend 
on the constants 
$\varepsilon\in (0,1)$ and $\delta\in (0,\infty)$, 
demonstrating that we have the strong theorem 
for the proposed framework of source encryption.
Information spectrum 
method \cite{Han98InfSpec} and 
a variant of Birkhoff-von Neumann 
theorem  play an important role 
in deriving those results.
We further prove the existence
of encryption/decryption schemes, which are universal in the
sense that they work effectively for any distributions 
of the plain text and those of the key used for 
the encryption.

\section{Source Encryption Framework }
\subsection{Preliminaries}
In this subsection, we show the basic notations 
and related consensus used in this paper.

\noindent \underline{\textit{Source 
of Information and Key:}}\/ 
We first define the source. 
Let $X$ be a random variable from a finite 
set $\mathcal{X}$. Let $\{X_t\}_{t=1}^\infty$ be a 
stationary discrete memoryless 
source (DMS) such that for each $t=1,2,\ldots,$
$X_t$ takes values in the finite set $\mathcal{X}$ 
and has the same distribution as that of $X$ 
denoted by $p_X=\{p_X(x)\}_{x \in \mathcal{X}}$.
The stationary DMS 
$\{X_t\}_{t=1}^\infty$ 
is specified with $p_{X}$.
We next define the key used in a secret-key 
cryptosystem. Let $K$ be a random variable 
from the same finite set $\mathcal{X}$, having the distribution $p_{K}=\{p_K(k)\}_{k \in \mathcal{X}}$.
Let $\{K_t\}_{t=1}^\infty$ 
be a stationary DMS specified with $p_K$.

\noindent 
\underline{\textit{Random Variables and Sequences:}} \/ We 
write the sequence of random variables of length 
$n$ from the information source as 
follows $\bm{X}  \coloneqq X_1 X_2 \ldots X_n$. 
Similarly, the strings of length $n$ in $\mathcal{X}^n$ 
are written as 
$\bm{x} \coloneqq x_1\cdots x_n\in\mathcal{X}^n$. 
For ${\bm{x}} \in \mathcal{X}^n$, $p_{\bm{X}}(\bm{x})$ 
stands for the probability of the occurrence of 
$\bm{x}$. When the information source 
is memoryless and specified by $p_X$, 
we have $p_{\bm{X}}(\bm{x})= \prod_{t=1}^n 
p_X(x_{t})$. In this case, we write $p_{\bm{X}}(\bm{x})$
as $p_X^n(\bm{x})$. Similar notations are used for 
other random variables and sequences.

\noindent \underline{\textit{Other Notation:}}~Without 
loss of generality, we assume that $\mathcal{X}$ 
is a finite field.
$\oplus$ denotes addition over the field, and $\ominus$ 
denotes subtraction over the field.
As an example, for any $a,b$ in the same finite field,
we have $a \ominus b = a \oplus (-b)$.
Moreover, in this paper, all logarithms are 
taken to base 2, and $\mathbb{N}$ denotes the 
set of natural numbers.

\subsection{Basic System Description}
In this subsection, we explain the basic system setting 
and basic adversarial model we consider in this paper.
The information source and the key are generated 
independently by different parties,
$\mathcal{S}_{\mathsf{gen}}$ 
and $\mathcal{K}_{\mathsf{gen}}$, respectively.

\noindent \underline{\it Source Coding without 
Encryption}: \/ The random source sequence $\bm{X}$ 
of length $n$ generated from $\mathcal{S}_{\mathsf{gen}}$ 
is sent to node $\mathsf{E}$.
We show the coding system 
in Fig.~\ref{fig:FFnoEnc}.

\begin{figure}[t]
  \centering
  \includegraphics[width=60mm]{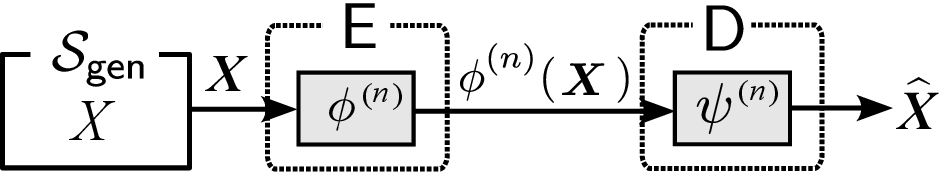}
  \caption{System model of source coding without encryption}
  \label{fig:FFnoEnc}
\end{figure}
The coding system is specified as follows.
\begin{itemize}
  \item[1)] {\it Encoding Process:}  
  At node $\mathsf{E}$, the encoder function $\phin : 
  \mathcal{X}^n \to \mathcal{M}^{(n)}$ observes 
  $\bm{X}$ to generate $\phinX$.
  Here, $\mathcal{M}^{(n)}=\{1,2,\cdots,|\mathcal{M}^{(n)}|\}$.

  \item[2)] {\it Transmission:}  
  Next, the encoded source $\phinX$ is sent to node $\mathsf{D}$ through a noiseless channel.

  \item[3)] {\it Decoding Process:}  
  At node $\mathsf{D}$, the decoder function $\psin : \mathcal{M}^{(n)} \to \mathcal{X}^n$ observes $\phinX$ to output $\widehat{\bm{X}}$, where
   $ \widehat{\bm{X}} \coloneqq \psin \circ \phinX.$
\end{itemize}

For the above $(\phin,\psin)$, define the set $\Dn$ of 
correct decoding by
\begin{align*}
 \Dn \coloneqq \{\bm{x} 
 \in \mathcal{X}^n:\psin \circ \phinx = \bm{x}\}.
\end{align*}

\noindent \underline{\it Source Encryption:}~The source 
sequence $\bm{X}$ generated 
from $\mathcal{S}_{\mathsf{gen}}$ 
and the key $\bm{K}$ generated 
from $\mathcal{K}_{\mathsf{gen}}$ 
are sent to node $\mathsf{L}$.
The details of the coding system 
are shown in Fig.~\ref{fig:FFwithEnc}.
\begin{figure}[t]
  \centering
  \includegraphics[width=60mm]{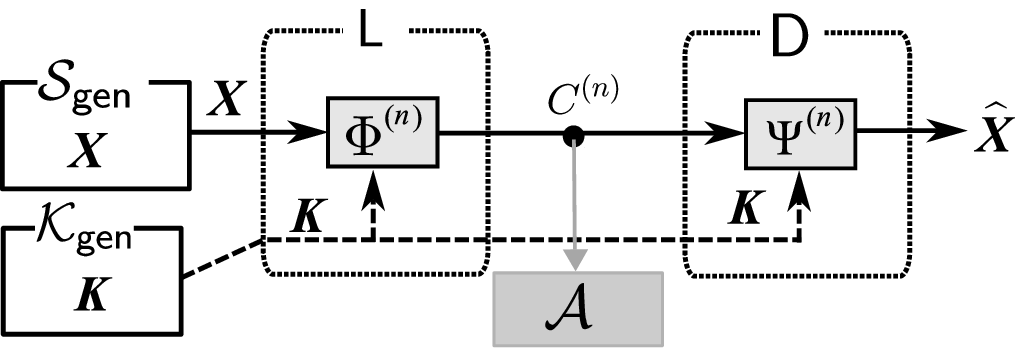}
  \caption{System model of source  encryption}
  \label{fig:FFwithEnc}
\end{figure}
The encrypted coding system is specified as follows.
\begin{itemize}
  \item[1)] {\it Source Processing:}  
  At node $\mathsf{L}$, $\bm{X}$ is encrypted with the key $\bm{K}$ using the encryption function
  $\Phin : \mathcal{X}^n \times \mathcal{X}^n \to \mathcal{C}^{(n)}$.
  The ciphertext of $\bm{X}$ is given by $C^{(n)} = \Phin(\bm{K}, \bm{X})$.
  On the encryption function $\Phin$, we use the following notation:
   $
    \Phin(\bm{K}, \bm{X})
    = \PhiKnX
    = \PhiXnK.$

  \item[2)]{\it Transmission:} 
  The ciphertext $C^{(n)}$ is sent to node $\mathsf{D}$ 
  through the public communication channel.
  Meanwhile, the key $\bm{K}$ is sent to $\mathsf{D}$ 
  through the private communication channel.

  \item[3)] {\it Sink Node Processing:}  
  At node $\mathsf{D}$, the ciphertext is decrypted using the key $\bm{K}$ through the corresponding decryption procedure
  $\Psin : \mathcal{X}^n \times \mathcal{C}^{(n)} \to \mathcal{X}^n$.
  Here we set $\widehat{\bm{X}} \coloneqq \Psin(\bm{K}, C^{(n)})$.
  On the decryption function $\Psin$, we use the following notation:
   $ \Psin(\bm{K}, C^{(n)})
    = \PsinK(C^{(n)})
    = \Psi_{C^{(n)}}^{(n)}(\bm{K}) $
\end{itemize}

Fix any $\bm{K} = \bm{k} \in \mathcal{X}^n$.
For this $\bm{K}$ and for $(\Phin, \Psin)$, we define the set $\mathcal{D}_{\bm{k}}^{(n)}$ of correct decoding by
\begin{align*}
  \mathcal{D}_{\bm{k}}^{(n)}
  \coloneqq
  \{\bm{x} \in \mathcal{X}^n : \Psin_{\bm{k}} \circ \Phiknx = \bm{x}\}.
\end{align*}

We require that the cryptosystem $(\Phin, \Psin)$ 
must satisfy the following condition.

\textbf{Condition}:~For each distributed source encryption 
system $(\Phin,\Psin)$,  there exists a source coding system 
$(\phin, \psin)$ such that for any $\bm{x} \in \mathcal{X}^{n}$ 
and for any $\bm{k} \in \mathcal{X}^{n}$, 
\begin{align*}
  \Psi_{\bm{k}}^{(n)} \circ \Phiknx = \psin \circ \phinx. 
\end{align*}
The above condition implies that 
$\Dn=\mathcal{D}_{\bm{k}}^{(n)}, 
\forall \bm{k} \in \mathcal{X}^n$.
We have the following properties on $\Dn$.

\begin{property} \label{per:Dn}
  For any $\bm{x}, \bm{x}^{\prime} \in \Dn$ with $\bm{x} \neq \bm{x}^{\prime}$, we have
  $\Phiknx \neq \Phin_{\bm{k}}(\bm{x}^{\prime})$.
\end{property}

Proof of Property~\ref{per:Dn} is given 
in Appendix \ref{prf:Dn}. 
From Property~\ref{per:Dn}, we have 
the following lemma.
  \begin{lemma} \label{lem:BirkoffvN}
  For any $c \in \mathcal{C}^{(n)}$, we have 
  \begin{align*}
    \sum_{\bm{x} \in \Dn} p_{C^{(n)}|\bm{X}}(c|\bm{x}) \leq 1.
  \end{align*}
\end{lemma}

Proof of this lemma is given 
in Appendix \ref{prf:Birkholf}. 
Lemma \ref{lem:BirkoffvN} 
can be regarded as an extension 
of the Birkhoff-von Neumann theorem~\cite{Iwamoto}.
Lemma~\ref{lem:BirkoffvN} plays an important 
role in proving the converse theorem.

\section{Problem Set Up and Main Results}
\subsection{Problem Set Up}
In this subsection, we introduce reliability and security 
criteria. 

\noindent \underline{\it Reliability Criterion:} \/
The decoding process is successful if $\widehat{\bm{X}}=\bm{X}$.
Hence the decoding error probability is given by
\begin{align*}
 &\Pr[\Psin(\bm{K}, \Phin(\bm{K}, \bm{X})) \neq \bm{X} ] 
  = \Pr[\PsinK \circ \PhiKnX \neq \bm{X} ]                   \\
  &=\Pr[\psin \circ \phinX \neq \bm{X}] =\Pr[\bm{X} \notin \Dn]
\end{align*}
Since the above quantity depends only on $(\phin, \psin)$, 
we write the error probability $p_{\mathrm{e}}$ of decoding as
\begin{align*}
  p_{\mathrm{e}}= p_{\mathrm{e}} (\phin,\psin |p_X^n)
  \coloneqq \Pr[\bm{X} \notin \Dn]
\end{align*}

\noindent \underline{\it Security Criterion:}\/
In communication, the  adversary $\mathcal{A}$ 
attempts to estimate the source sequence 
$\bm{X} \in \mathcal{X}^n$ from the ciphertext $C^{(n)}$.
Therefore, we define the mutual information(MI) 
between $\bm{X}$ and $C^{(n)}$ as
$\MI
= \MI(\Phin\mid p_X^n,p_K^n$ $)
\coloneqq I(C^{(n)}; \bm{X}),$
and adopt
it as the security criterion.

\begin{definition} \label{def:admissible}
  We fix some positive constant $\delta_0$. 
  For a fixed pair $(\varepsilon, \delta)\in (0,1) \times (0,\delta_0]$,
  a quantity $R$ is $(\varepsilon,\delta)$-admissible 
  if $\exists \{(\Phin,\Psin)\}^{\infty}_{n=1}$ such that
  $\forall \gamma >0$,$\exists n_0=n_0(\gamma) \in \mathbb{N}$, $\forall n\geq n_0$, 
  \begin{align*}
     & \frac{1}{n} \log | \mathcal{C}^{(n)}| \leq R+ \gamma,                             \\
     & p_{\mathrm{e}}(\phin,\psin|p_X^n) \leq \varepsilon,
     ~\Delta^{(n)}_{\mathrm{MI}}(\Phin|p_X^n,p_K^n) \leq \delta.
  \end{align*}
\end{definition}

\begin{definition}
  Let $R_{\inf}(\varepsilon,\delta |p_X, p_K)$ denote 
  the infimum of all $(\varepsilon,\delta)$-admissible rates $R$.
  Furthermore, we define
  \begin{align*}
    R_{\inf}(p_X, p_K) \coloneqq \sup_{\substack{(\varepsilon, \delta) 
    \in  (0,1) \\\times (0,\delta_0] }} 
    R_{\inf}(\varepsilon,\delta |p_X, p_K).
  \end{align*}
\end{definition}

\subsection{Main Results}
We first state several definitions 
to describe the results.
Let $\mathcal{P}(\mathcal{X})$ denote the set of all probability distributions on $\mathcal{X}$.
For $R \geq 0$ and 
$(p_X,p_K)\in \mathcal{P}^2(\mathcal{X})$, 
we define the following functions:
\begin{align*}
& E(R|p_X) \coloneqq \min_{\substack{ P
  \in \mathcal{P}(\mathcal{X}): \\ R \leq H(P)}} 
  D(P||p_X),
\\ 
&F(R|p_K)
 \coloneqq \min_{
 P \in \mathcal{P}(\mathcal{X})}
\left\{
[H(P)-R]^{+}+D(P||p_K)
\right\}.
\end{align*}
Here $[a]^+=\max\{0,a\}$.
For the functions $E(R|p_X)$ and $F(R|p_K)$, 
we have the following property.
\begin{property}\label{per:F(R|pK)}
The two functions $E_{\empty}(R|p_X)$ and $F(R|p_K)$ take positive values
if and only if $H(X) \! <\! R\! <\! H(K)$.
\end{property}

We set 
\begin{align}
  \gamma_n &\coloneqq 
  \frac{1}{n}\{|\mathcal{X}|\log(n+1)+\log|{\cal X}|+1\},
  R_n  \coloneq R+\gamma_n.
 \label{eqn:ChoicRate}  
\end{align}
Note that $\gamma_n$ vanishes as $n \to \infty$.
We have the following theorem.
\begin{theorem}[Direct Coding Theorem]\label{thm:Directproof}
$\forall R>0$, 
$\exists \{(\Phin,$ $\Psin)\}_{n=1}^{\infty}$ 
such that $\forall n\in \mathbb{N}$ 
and  $\forall (p_X,p_K) 
\in {\cal P}^2({\cal X})$,  
\begin{align}
\frac{1}{n}\log |{\cal C}^{(n)}|&\in
\left[
 R_n-\frac{1}{n} 
 \log |\mathcal{X}|, R_n \right],
\label{eqn:DirectThRate}\\
     p_{\mathrm{e}}(\phin,\psin|p_X^n) 
     &\leq (n+1)^{|\mathcal{X}|}2^{-nE(R|p_X)}, 
\label{eqn:ErrUbd} \\
     \MI(\Phin|p_X^n,p_K^n) 
      &\leq (2R_n+1)|\mathcal{X}|
       (n+1)^{4|\mathcal{X}|}
\notag\\      
      &\quad \times 2^{-nF(R|p_K)}.
\label{eqn:SecUbd}
  \end{align}
\end{theorem}

Proof of Theorem~\ref{thm:Directproof} 
is given in the Section \ref{prf:thm_Direct}. 
From Theorem~\ref{thm:Directproof} 
and Property \ref{per:F(R|pK)}, we have 
the following corollary: 
\begin{corollary}\label{cor:directThA}
$\forall R>0$, $\exists \{(\Phin, \Psin)$ $\}_{n=1}^{\infty}$ 
such that $\forall (p_X,$ $p_K)\in {\cal P}^2({\cal X})$ 
with $H(X)<R<H(K)$, 
\begin{align}
\left.
\begin{array}{l}
{\displaystyle \lim_{n\to\infty}} 
(1/n)\log |{\cal C}^{(n)}|=R, 
\vspace{1mm}\\
{\displaystyle \liminf_{n\to\infty}}(-1/n)
\log p_{\mathrm{e}}(\phin,\psin|p_X^n)
\\
\geq E(R|p_X)>0,
\vspace{1mm}\\
{\displaystyle \liminf_{n\to\infty}}(-1/n)
\log \MI(\Phin|p_X^n,p_K^n)
\\ 
\geq F(R|p_K)>0.
\end{array}
\right\}
\label{eqn:ThOneBound}
\end{align}
\end{corollary}
By Corollary \ref{cor:directThA}, 
under $H(X)<R<H(K)$,  
we have the followings: 
\begin{itemize}
\item[1.] On the reliability, 
$p_{\rm e}(\phi^{(n)},\psi^{(n)}|p_{X}^n)$ 
vanishes exponentially as $n\to\infty$, and its 
exponent is lower bounded by 
$E(R|p_{X})$.  
\item[2.] On the security, 
$\Delta_{\rm MI}^{(n)}
(\Phi^{(n)}$$|p_{X}^n,p_{K}^n)$ vanishes exponentially 
as $n\to \infty$, and its exponent 
is lower bounded by $F(R|p_{K})$.
\item[3.] The code that attains the pair 
$(E($$R|p_{X})$, $F(R|p_{X}))$ 
of exponent functions is the universal 
code that depends only on $R$ not 
on the value of the pair of the 
distributions  
$(p_{X},p_K)\in {\cal P}^2({\cal X})$.
\end{itemize}


Here, we define the following quantity.
\begin{align*}\label{def:Rast2}
 R^{\ast}(p_X,p_K)=
   \begin{cases}
    H(X)      &\mbox{ if } H(X) < H(K), \\
    +\infty &\mbox{ otherwise}.
  \end{cases}
  \stepcounter{equation}\tag{\theequation}
\end{align*}
From the definition of $R^{\ast}(p_X, p_K)$, 
Corollary \ref{cor:directThA}, 
and Property~\ref{per:F(R|pK)}, we have 
the following corollary: 
\begin{corollary}\label{cor:directTh}
For each $(\varepsilon,\delta)
 \in (0,1)\times (0,\delta_0]$, we have
\begin{align*}
R_{\inf}(\varepsilon,\delta |p_X,p_K)
\leq R_{\inf}(p_X,p_K)\leq R^{\ast}(p_X, p_K).
\end{align*}
\end{corollary}

\textit{Proof:}
Choose $R$ 
such that $H(X) < R < H(K)$. 
Since $H(X)<H(K)$, 
this choice of $R>0$ is possible.   
Then, for any 
$\tau \in \bigl(0,\min\{{R-[H(X)], 
                   H(K)-R}\bigr]\}$, 
we have the following inequality:
\begin{align}\label{AA}
H(X) +\tau\leq R\leq H(K)-\tau.
\end{align}
Then it follows from 
Corollary \ref{cor:directThA}
that $\exists 
\{(\Phin,\Psin)\}_{n =1}^{\infty}$ 
such that we have 
the bound (\ref{eqn:ThOneBound}) 
in this corollary, implying 
that for any 
$(\varepsilon,\delta)$ 
$\in (0,1)\times (0,\delta_0]$, 
every $R$ satisfying \eqref{AA} 
is $(\varepsilon,\delta)$-admissible.
Since $\tau > 0$ in \eqref{AA} 
can be chosen arbitrarily small, it follows 
that under $H(X)<H(K)$, every $R$ satisfying 
$H(X)\leq R\leq H(K)$
is $(\varepsilon,\delta)$-admissible.
Combining this with the definition of 
$R^{\ast}(p_X,p_K)$ given 
by \eqref{def:Rast2}, we have that
for any $(\varepsilon,\delta)$ 
$\in (0,1) \times (0,\delta_0]$, 
$R^{\ast}(p_X,p_K)$ 
is $(\varepsilon,\delta)$-admissible.
Since $(\varepsilon,\delta)$ can arbitrary 
be close to $(0,0)$, we conclude 
that $R^{\ast}(p_X,p_K)$ $\geq R_{\rm inf}(p_X,p_K)$.   \hfill \IEEEQED

We next describe a result on the 
converse coding theorem. 
To this end we set    
\begin{align*}\label{def:Rast3}
   R^{\star}(p_X, p_K) 
   =
  \begin{cases}
    H(X) &\mbox{ if } H(X) \leq H(K), \\
    +\infty        &\mbox{ otherwise}.
  \end{cases}
  \stepcounter{equation}\tag{\theequation}
\end{align*}

We have the following result on a lower bound of
$R_{\inf}(\varepsilon,\delta\mid p_X, p_K)$ 
for $(\varepsilon,\delta)
\in (0,1)\times(0,\delta_0]$:
\begin{theorem}[Converse Coding Theorem]
\label{th:StConvTh}
For each $(\varepsilon, \delta)\in (0,1) \times (0,\delta_0]$,we have
\begin{align*}
 R^{\star}(p_X, p_K) 
 \leq R_{\inf}(\varepsilon,\delta \mid p_X, p_K) 
 \leq R_{\inf}(p_X, p_K).
  \end{align*}
\end{theorem}

Proof of Theorem~\ref{th:StConvTh} is given in Section~\ref{prf:thm_Conv}.
Combining Corollary \ref{cor:directTh} and Theorem \ref{th:StConvTh}, we have the following result.
\begin{theorem}\label{th:CodingTh_FF}
Consider the case of $H(X)<H(K)$. In this case 
we have  
$
R^{\star}(p_X, p_K)=R^{\ast}(p_X, p_K)=H(X).
$
Hence for each $(\varepsilon, \delta) \in (0,1) \times (0,\delta_0]$, we have
\begin{align*}
R_{\inf}(p_X, p_K)
= R_{\inf}(\varepsilon,\delta \mid p_X, p_K)
=H(X).
\end{align*}
\end{theorem}

Theorem \ref{th:CodingTh_FF} implies that we have the strong converse property for
$R_{\inf}(\varepsilon,\delta\mid p_X,p_K),$
$(\varepsilon, \delta)\in (0,1)\times (0,\delta_0]$.

\section{Proofs of Theorems \ref{thm:Directproof} and \ref{th:StConvTh}}
In this section, we prove Theorems \ref{thm:Directproof} 
and \ref{th:StConvTh}. 
Theorem \ref{thm:Directproof} corresponds 
to the direct coding theorem.
Theorem \ref{th:StConvTh} corresponds
to the converse coding theorem. 

\subsection{Proof of Theorem 
\ref{thm:Directproof}}\label{prf:thm_Direct}
In this subsection, we prove Theorem \ref{thm:Directproof}. 
We first introduce the definition of 
types used in the proof.

\begin{definition} 
Let $N(x|\bm{x})$ denote the number of occurrences of $x \in \mathcal{X}$ appearing in sequence $\bm{x} = x_1x_2\ldots x_n$.
The probability distribution 
$P_{\bm{x}}$ on ${\cal X}$ defined by
\begin{align*}
  P_{\bm{x}}\coloneqq \left\{P_{\bm{x}}(x)\right\}_{x \in \mathcal{X}} = \left\{\frac{1}{n}N(x|\bm{x})\right\}_{x \in \mathcal{X}}
\end{align*}
is called the type of $\bm{x}$ on ${\cal X}$.
The set that consists of all the
types on
${\cal X}$ is denoted by 
$\mathcal{P}_n({\cal X})$.
For each $P \in \mathcal{P}_n({\cal X})$, we define
$
  T^n(P) \coloneqq \{\bm{x} \in \mathcal{X}^n|P_{\bm{x}}=P\}
$  
as the set of all $\bm{x}$ having type $P$.
\end{definition}


\noindent
\underline{\it Choices of Some Parameters:}\/ For 
a given rate $R>0$, choose $m\in \mathbb{N}$ 
such that
\begin{align}
  m = 
  \left\lfloor\frac{nR_n}{\log |\mathcal{X}|}\right\rfloor, \: R_n=R+\gamma_n.
\label{eqn:ChoosEm}
\end{align}
The choice (\ref{eqn:ChoosEm}) of $m$ implies 
the following:
\begin{align}
&  R_n-\frac{1}{n} 
 \log |\mathcal{X}|
  \leq 
  \frac{m}{n} \log |\mathcal{X}|
   \leq R_n.
\label{eqn:Rconstraint}   
\end{align}

\noindent 
\underline{\textit{Construction of 
$(\phin,\psin)$:}} \/ We first present 
the construction of the code to be used.
We define the union of sets of sequences having types $P$ whose entropy values are smaller than $R$ as follows.
\begin{align}\label{def:CRn}
  \mathcal{C}^{n}(R) 
  \coloneqq \bigcup_{\substack{P \in \mathcal{P}_n({\cal X}): \\ R > H(P)}} T^n(P).
\end{align}
For $\mathcal{C}^{n}(R)$ 
defined above, we have 
the following lemma: 
\begin{lemma} \label{lem:upbCR}
$\forall \gamma > 0$ and $\forall n 
\in \mathbb{N}$, we have 
\begin{align}
&|\mathcal{C}^{n}(R)| 
\leq (n+1)^{|\mathcal{X}|}2^{nR}
\label{eqn:calCupb}\\
& 
\leq \frac{1}{2}|\mathcal{X}|^m
\leq|\mathcal{X}^m|-1.
\label{eqn:calCupc}
\end{align}
\end{lemma}

Proof of Lemma \ref{lem:upbCR} is given in Section~\ref{prf:pos_Kobayashi}.
We construct $(\phin,\psin)$ based on 
this lemma. Let $x_0^m$ be an arbitrary 
element of $\mathcal{X}^m$. 
Due to Lemma \ref{lem:upbCR}, 
there exists a one-to-one mapping 
$\widetilde{\phi}^{(n)}:
\mathcal{C}_{\empty}^{n}(R) 
\to \mathcal{X}^m-\{x_0^m\}$.
Using this mapping, we define the mapping 
$\phin:\mathcal{X}^n\to \mathcal{X}^m$ 
as follows.
\begin{align}\label{def:phi_direct}
  \phinx
  \coloneqq\left\{\begin{array}{cl}
    \widetilde{\phi}^{(n)}(\bm{x}) ,
     &\mbox{ if }\bm{x} \in \mathcal{C}_{\empty}^{n}(R),
    \vspace{0.2cm} \\
    x_0^m, &\mbox{ if }\bm{x}\notin \mathcal{C}_{\empty}^{n}(R).
  \end{array}
  \right.
\end{align}
Moreover, $\psin$ is defined as a mapping such that 
based on the one-to-one mapping
$\widetilde{\phi}^{(n)}:\mathcal{C}_{\empty}^{n}(R) 
\to \mathcal{X}^m$,
each element of $\widetilde{\phi}^{(n)}(\mathcal{C}_{\empty}^{n}(R))$ 
is associated with the corresponding element 
of $\mathcal{C}_{\empty}^{n}(R)$.
For sequences belonging to 
$\mathcal{X}^n-\mathcal{C}_{\empty}^{n}(R)$,
the decoder maps them to $x_0^m$.
The encoding process is shown 
in Fig. \ref{fig:CRn_shazou}.
\begin{figure}[t]
  \centering
  \includegraphics[width=75mm]{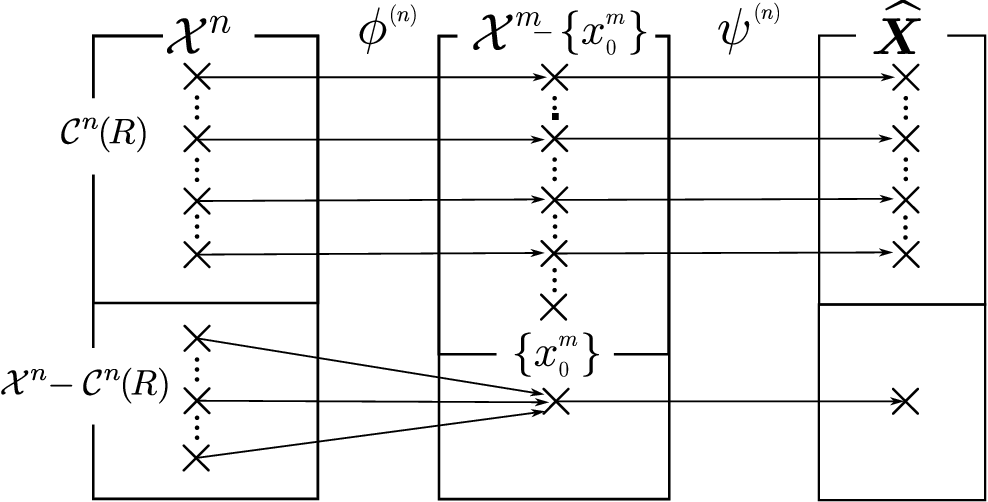}
  \caption{Encoding process 
  based on $\mathcal{C}^{n}(R)$}
  \label{fig:CRn_shazou}
\vspace*{-2mm}
\end{figure}
\begin{figure}[t]
  \centering
  \includegraphics[width=80mm]{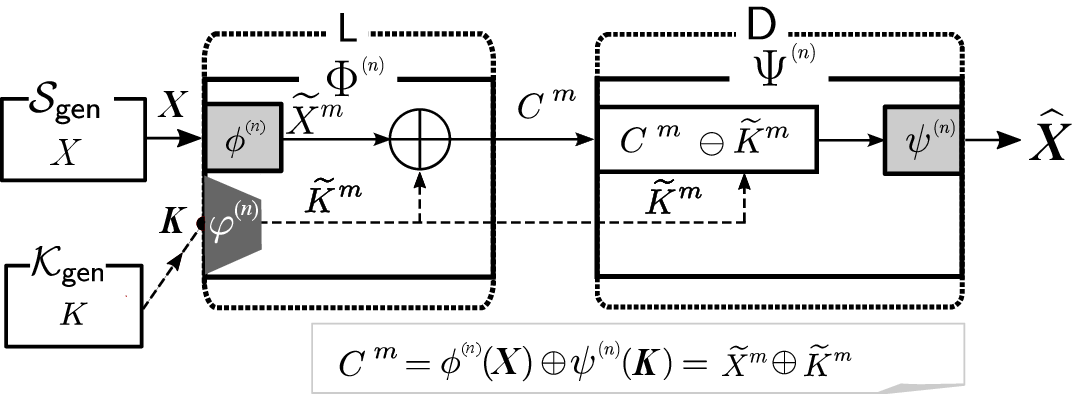}
  \vspace*{5mm}
  \caption{Constructions of $(\Phin,\Psin)$}
  \label{fig:UniversalCoding}
\end{figure}

\noindent \underline{\textit{Affine Encoder:}} \/ 
 Let $A$ be a matrix with $n$
rows and $m$ columns. Entries of $A$ take values 
in $\mathcal{X}$ . Fix $b^m \in \mathcal{X}^m$.
We  define an affine encoder 
$\varphi^{(n)}:\mathcal{X}^n\to\mathcal{X}^m$ 
determined by the matrix $A$ 
and the vector $b^m$ as follows.
\begin{align}
  \varphi^{(n)}(\bm{k}) \coloneqq \bm{k} A\oplus b^{m}.
\end{align}

\noindent \underline{\textit{Encryption Scheme:}}\/ 
Using the code $(\phin,\psin)$ and the affine 
encoder $\varphi^{(n)}$, we define the encoding
scheme as shown in Fig.~\ref{fig:UniversalCoding}.

The constructions of $\Phin$ and $\Psin$ are as follows.
\begin{itemize}
  \item[1)] 
  Define $\Phin:\mathcal{X}^{n} \times \mathcal{X}^{n} 
 \to \mathcal{X}^m$ by
 \begin{align*}
 \Phin(\bm{k},\bm{x})=\varphi^{(n)}(\bm{k})\oplus \phinx 
 \quad \mathrm{for}~\bm{k},\bm{x} \in \mathcal{X}^n.
 \end{align*}
Let $C^m = \Phin(\bm{X},\bm{K})$, 
    $\widetilde{X}^m=\phinX$, 
and $\widetilde{K}^m=\varphi^{(n)}(\bm{K})$.
Then we have 
$C^m=\widetilde{X}^m \oplus \widetilde{K}^m$. 
The ciphertext $C^m$ is sent to the public 
communication channel.
 \item[2)] 
 $\Psin$ receives the ciphertext $C^m$ $=\widetilde{X}^m \oplus \widetilde{K}^m$ and the key $\bm{K}$, respectively, through public and private channels. 
Using $\varphi^{(n)}$, $\Psin$ first encodes 
$\bm{K}$ into $\widetilde{K}^m = \varphi^{(n)}(\bm{K})$.
$\Psin$ next subtracts $\widetilde{K}^m$ from $C^m$ 
to obtain $\widetilde{X}^m = \phinX$.
Finally, $\Psin$ outputs $\widehat{\bm{X}}$ by applying the decoder $\psin$ to $\widetilde{X}^m$.
\end{itemize}

\noindent \underline{\textit{Reliability Analysis:}} \/ 
We have the following proposition related to the error probability of the encryption scheme.
\begin{proposition} \label{pos:Kobayashi}
  For given $ R > 0$,
  $\exists \{(\phin,\psin)$ $\}_{n=1}^{\infty}$ with the choice of $m$ satisfying (\ref{eqn:ChoosEm})
such that $\forall n \in \mathbb{N}$ and 
$\forall p_X\in {\cal P}({\cal X})$, 
 \begin{align*}
  &p_{\mathrm{e}}(\phin,\psin|p_X^n) \leq 
    \Pr\{\mathcal{X}^n -\mathcal{C}^{n} 
    (R)\}\notag\\
    &\leq (n+1)^{|\mathcal{X}|}2^{-nE(R|p_X)}.
  \end{align*}
\end{proposition}

Proof of Proposition~\ref{pos:Kobayashi} is given in Section~\ref{prf:pos_Kobayashi}.

\noindent \underline{\textit{Security Analysis:}} \/ 
With respect to the mutual information of 
the encryption scheme, we have the following 
proposition: 
\begin{proposition}\label{pos:upb_Max-MI}
 $\exists\{\varphi^{(n)}\}_{n=1}^{\infty}$ 
with the choice of $m$ satisfying (\ref{eqn:ChoosEm})
such that $\forall n \in \mathbb{N}$ and 
$\forall p_K\in {\cal P}({\cal X})$, 
\begin{align*}
&   \Delta^{(n)}_{\mathrm{MI}}(\Phin|p_X^n,p_K^n) 
\leq 
m\log |\mathcal{X}|-H(\widetilde{K}^m)
\\
&\leq \left(R_n+\frac{1}{2}\right)(n+1)^{3|{\cal X}|}
     2^{-n[F(R|p_K)-\gamma_n]}
\\
&=(2R_n+1)|{\cal X}|(n+1)^{4|{\cal X}|}
     2^{-nF(R|p_K)}.    
  \end{align*}
\end{proposition}

Proof of Proposition~\ref{pos:upb_Max-MI} is given in Section~\ref{prf:pos_upb_Max-MI}.

\textit{Proof of Theorem~\ref{thm:Directproof}:}
For each $n \in \mathbb{N}$, we take $(\Phin,\Psin)$, 
the construction of which we have described in the explanation on the encryption scheme. 
For this $(\Phin,\Psin)$, we have 
$$
\frac{1}{n}\log |{\cal C}^{(n)}|=
\frac{1}{n}\log |{\cal X}^m|\MIn{a} 
\left[R_n- \frac{1}{n}\log |{\cal X}|, R_n\right].
$$
Step (a) follows from (\ref{eqn:Rconstraint}). 
Hence we have (\ref{eqn:DirectThRate}) 
in Theorem \ref{thm:Directproof}.
By Propositions~\ref{pos:Kobayashi} 
and~\ref{pos:upb_Max-MI}, respectively, we have 
(\ref{eqn:ErrUbd}) and (\ref{eqn:SecUbd})
in Theorem~\ref{thm:Directproof}.
\hfill \IEEEQED

\subsection{Proof of Theorem~\ref{th:StConvTh}}\label{prf:thm_Conv}
In this subsection, we prove Theorem~\ref{th:StConvTh}.
To prove this theorem it suffices to show that under the assumption that $R$ is $(\varepsilon,\delta)$-admissible we have $R \geq R^{\star}(p_X,p_K)$.
From (\ref{def:Rast2}), $R \geq R^{\star}(p_X,p_K)$ 
is equivalent to the following:
\begin{align}
  &R \geq H(X),\:
  H(K) \geq H(X).\label{RKgeqX}
\end{align}
\newcommand{\commentoutE}{
In the following we prove Theorem~\ref{th:StConvTh}.

\noindent \underline{Proof of Theorem\ref{th:StConvTh} 
Using Change of Measure}\par
  We present a proposition necessary for proof 
  of Theorem \ref{th:StConvTh}. 
  \begin{proposition}\label{pos:HX_HK_com}
 Let $\displaystyle p_{\min} \coloneqq \min_{x \in \mathcal{X}} p_X(x)$. 
 If $R$ is $(\varepsilon,\delta)$-admissible, 
 we have the following:
  \begin{align*}
     H(K) \geq &H(X) - 3\biggl[\frac{1}{2}\left(\log \frac{1}{p_{\min}}\right)^2\log|\mathcal{X}|\biggr]^{\frac{1}{3}}\\
     &\times \biggl[\frac{1}{n}
      \log \frac{1}{1 - \varepsilon}\biggr]^{\frac{1}{3}}
      -\frac{\varepsilon}{n(1-\delta)}.
      \end{align*}
    \end{proposition}

  We can prove Proposition \ref{pos:HX_HK_com} by using an argument of change of measure developed by Tyagi and Watanabe \cite{Change_of_measure}.
  Proof of Proposition~\ref{pos:HX_HK_com} is given in Section~\ref{prf:HX_HK_com}.
}
We define the following:
\begin{align*}
&\widetilde{\mathcal{A}}^{(n)}_{\gamma} 
 \coloneqq \left\{\bm{x}:\frac{1}{n} \log \frac{1}{p_X^n(\bm{x})} 
 \geq H(X) - \gamma \right\},
 \\
&
\widetilde{\mathcal{B}}^{(n)}_{\gamma} 
 \coloneqq \widetilde{\mathcal{A}}^{(n)}_{\gamma} \cap \Dn,\\
  &\nu_n(\gamma) \coloneqq 
 p_X^n\left((\widetilde{\mathcal{A}}^{(n)}_{\gamma})^c\right), \quad 
  \widetilde{\nu}_n(\gamma,\varepsilon)
  \coloneqq \nu_n(\gamma) + \varepsilon, \\
  &\zeta_n(\gamma, \varepsilon, \delta) \coloneqq \frac{1}{n} 
     \left[ \frac{\varepsilon}{1-\widetilde{\nu}_n(\gamma, \varepsilon)} 
    + \log\frac{1}{1-\widetilde{\nu}_n(\gamma,\varepsilon)}\right]. 
\end{align*}
By the large deviation principle, for fixed $\gamma > 0$, 
$\nu_n(\gamma)$ goes to zero exponentially as $n \to \infty$.
Moreover, we have the following upper bound 
for $\widetilde{\mathcal{B}}^{(n)}_{\gamma}$:
\begin{align}\label{upb:Pr_B^c}
  p_X^n\left((\widetilde{\mathcal{B}}^{(n)}_{\gamma})^c\right) 
  &\leq   
  p_X^n\left((\widetilde{\mathcal{A}}^{(n)}_{\gamma})^c\right) 
  + p_X^n\left((\Dn)^c\right)\notag\\
  &\MLeq{a} \widetilde{\nu}_n(\gamma, \varepsilon).
\end{align}
Step (a) follows from $p_X^n\left((\Dn)^c\right) 
= p_{\mathrm{e}}(\phin,\psin|p_X^n)\leq\varepsilon$.
Then, we have the following proposition.
    \begin{proposition}\label{pos:Oohama25_1}
      Fix any $(\varepsilon, \delta) \in 
      (0,1) \times (0, \delta_0]$.
      Then, for all $\gamma >0$, there 
      exists $n_0 = n_0(\gamma)$ 
      such that for all $n \geq n_0$, we have the following:
      \begin{align*}
        H(K) \geq H(X) + \gamma + \zeta_n(\gamma, \delta, \varepsilon).
      \end{align*}
    \end{proposition}

Proof of Proposition~\ref{pos:Oohama25_1} 
is given in Section~\ref{prf:pos_Oohama25_1}.
Under these preparations, we prove 
Theorem~\ref{th:StConvTh}.
    \newcommand{\commentout}{
 \begin{remark}
   The method using an extension of 
   the {\rm Birkhoff--von Neumann} 
   theorem requires assuming the universality of the key, 
   and hence it is necessary to 
   require $\Dn=\mathcal{D}_{\bm{k}}^{(n)}$.
      \end{remark}
    }

\textit{Proof of Theorem \ref{th:StConvTh}:} 
We assume $R$ is $(\varepsilon,\delta)$-admissible. 
Then there exists $\{(\Phin,\Psin)\}^{\infty}_{n=1}$ 
such that $\forall \gamma >0$,
$\exists n_0=n_0(\gamma) \in \mathbb{N}$, 
$\forall n\geq n_0$, we have the following:
\begin{align*}
  &\frac{1}{n} \log |{\mathcal{C}}^{(n)}| \leq R+ \gamma,
  ~p_{\mathrm{e}}(\phin,\psin|p_X^n) \leq \varepsilon,\\
  &\Delta^{(n)}_{\mathrm{MI}}(\Phin|p_X^n,p_K^n)
   =I(C^{(n)};\bm{X})
  \leq \delta. 
\end{align*}
We first prove the first bound in (\ref{RKgeqX}) 
in Theorem \ref{th:StConvTh}.
Assume that $R < H(X)$. Since we have 
the strong converse theorem 
for the source coding of discrete memoryless sources, we 
have that    
for any $\{(\phin,\psin)\}^{\infty}_{n=1}$,
for all $\tau \in (0,1)$ and all $\gamma > 0$, 
there exists $n_1 = n_1(\tau,\gamma) > 0$ such 
that for all $n \geq n_1$, we have the following:
\begin{align} \label{gyaku5}
  {p_{\mathrm{e}}(\phin,\psin|p_X^n)} \geq 1-\tau.
\end{align}
In (\ref{gyaku5}), by choosing $\tau = 1-\sqrt{\varepsilon}$, we have the following:
\begin{align} \label{gyaku6}
  {p_{\mathrm{e}}(\phin,\psin|p_X^n)} \geq \sqrt{\varepsilon} \MG{a} \varepsilon.
\end{align}
Step (a) follows from $\varepsilon = (1-\tau)^2 \in (0,1)$. 
Equation~(\ref{gyaku6}) contradicts $p_{\mathrm{e}}(\phin,\psin|p_X^n) \leq \varepsilon$.
Therefore, if $R$ is $(\varepsilon,\delta)$-admissible, we have $R \geq H(X)$.

\newcommand{\commentoutF}{
Next, we show (\ref{KgeqX}).
From Proposition~\ref{pos:HX_HK_com}, we have the following:
      \begin{align}\label{gyaku7}
        H(K) \geq &H(X) - 3\biggl[\frac{1}{2}\left(\log \frac{1}{p_{\min}}\right)^2\log|\mathcal{X}|\biggr]^{\frac{1}{3}}\notag\\
        &\times \biggl[\frac{1}{n}\log \frac{1}{1 - \varepsilon}\biggr]^{\frac{1}{3}}- \frac{\varepsilon}{n(1-\delta)}.
      \end{align}
Taking $n \to \infty$ in (\ref{gyaku7}), we have $H(K) \geq H(X)$.
Since $R$ was arbitrary, it follows that $R_{\inf}(p_X,p_K) \geq R^{\ast}(p_X,p_K)$.
\QED
}

We next show the second bound in (\ref{RKgeqX}). From Proposition~\ref{pos:Oohama25_1}, 
we have
\begin{align}\label{gyaku13}
        H(K) \geq H(X) + \gamma + \zeta_n(\gamma, \delta, \varepsilon).
      \end{align}
Taking $n \to \infty$ in (\ref{gyaku13}), we have 
$
  H(K) \geq H(X) + \gamma.$
Since $\gamma$ can be made arbitrarily small, we have $H(K) \geq H(X)$.
\hfill\IEEEQED
 
\section{Proofs of Lemma \ref{lem:upbCR}, Propositions \ref{pos:Kobayashi} and \ref{pos:upb_Max-MI}}
\subsection{Proofs of Lemma \ref{lem:upbCR} and Proposition~\ref{pos:Kobayashi}}
\label{prf:pos_Kobayashi}
In this subsection, we prove 
Lemma \ref{lem:upbCR} and  Proposition~\ref{pos:Kobayashi}.
We present below the properties of types necessary for the proofs of those two results. 
Proposition~\ref{pos:Kobayashi}.

\begin{lemma}\label{lem:type}
For types, we have the following properties.
\begin{itemize}
\item[{\rm a)}] $|\mathcal{P}_n({\cal X})| \leq (n + 1)^{|\mathcal{X}|-1} \leq (n + 1)^{|\mathcal{X}|}$.

\newcommand{\OmiTTz}{
\item[{\rm b)}] The probability that a stationary memoryless source specified by distribution $p_X$ outputs a sequence $\bm{x}$ of type $P$ is given as follows:
\begin{align*}
  p_X^n(\bm{x}) = 2^{-n(H(P)+ D(P||p_X))}.
\end{align*}
\item[{\rm c)}] In {\rm b)}, if the type of the sequence $\bm{x}$ is $P$, then we have the following:
\begin{align*}
  P^n(\bm{x}) = 2^{-nH(P)}.
\end{align*}
\item[{\rm d)}] If $P \in \mathca{P}_n({\cal X})$ is a type, then for all types $P' \in \mathca{P}_n({\cal X})$, we have the following:
\begin{align*}
  P^n(T^n(P)) \geq P^n(T^n(P')).
\end{align*}
}
\item[{\rm b)}] For any type $P \in \mathcal{P}^n$, 
we have the following:
\begin{align*}
  &\frac{1}{(n+1)^{|\mathcal{X}|-1}} 
  \leq \frac{|T^n(P)|}{2^{nH(P)}} \leq 1.
\end{align*}
\item[{\rm c)}] $\forall P \in \mathcal{P}_n({\cal X})$ and $\forall p_X \in \mathcal{P}({\cal X})$, 
we have the following:
\begin{align*}
  &\frac{1}{(n+1)^{|\mathcal{X}|-1}} 
   \leq \frac{p_X^n(T^n(P))} {2^{-nD(P||p_X)}}
   \leq 1. 
\end{align*}
\end{itemize}
\end{lemma}

Proof of this lemma 
is given in Han and Kobayashi~\cite{HanKobayashi}.
Under Lemma~\ref{lem:type}, we prove 
Lemma \ref{lem:upbCR} and 
Proposition~\ref{pos:Kobayashi}.

\textit{Proof of Lemma~\ref{lem:upbCR}:}
We first prove the bound (\ref{eqn:calCupb}) 
in Lemma \ref{lem:upbCR}.
On upper bounds of 
$|\mathcal{C}^n(R)|$, we have
the following:
\begin{align*} 
  &|\mathcal{C}^n(R)|=
 \sum_{\substack{P \in \mathcal{P}_n({\cal X}):\\
   R  > H(P)}} |T^n(P)|
  \MLeq{a} 
  \sum_{\substack{P \in \mathcal{P}_n({\cal X}):\\
     R > H(P)}} \hspace*{-10pt} 2^{nH(P)} 
  \\   
   & \MLeq{b} |\mathcal{P}_n({\cal X})|2^{nR}
  \MLeq{c} (n+1)^{|\mathcal{X}|} 2^{nR}.
\end{align*}
Step (a) follows from Lemma~\ref{lem:type}~part~b).
Step (b) follows from $H(P)<R$.
Step (c) follows from 
Lemma~\ref{lem:type}~part~a).
We next prove the bound 
(\ref{eqn:calCupc}). 
We further evaluate upper bounds of $|\mathcal{C}^n(R)|$ 
to obtain the following chain of inequalities:  
\begin{align*}
 &|\mathcal{C}^n(R)|
  \MLeq{a}(n+1)^{|\mathcal{X}|}2^{nR}
  =2^{-n\gamma_n}(n+1)^{|\mathcal{X}|}2^{nR_n} 
 \\
 & \MLeq{b} 2^{-n\gamma_n} (n+1)^{|\mathcal{X}|} 
 |\mathcal{X}|^{m+1}
 \MEq{c} \frac{1}{2}|\mathcal{X}|^m
 \leq|\mathcal{X}^m|-1.
\end{align*}
Step (a) follows from 
the bound (\ref{eqn:calCupb}) 
we have now established. 
Step (b) follows from that 
the bound (\ref{eqn:Rconstraint})
implies $2^{nR_n} \leq |\mathcal{X}|^{m+1}$. 
Step (c) follows from   
$2^{-n\gamma_n} =\left[2(n+1)^{|\mathcal{X}|}
|\mathcal{X}|\right]^{-1}$.
\hfill\IEEEQED

\textit{Proof of Proposition~\ref{pos:Kobayashi}:}
For $(\phin, \psin)$ defined in (\ref{def:phi_direct}), 
we derive an upper bound of the error probability 
$p_{\mathrm{e}}($ $\phin,\psin|p_X^n)$.
From the definition of $(\phin, \psin)$ , the error probability $p_{\mathrm{e}}(\phin,\psin|p_X^n)$ is upper bounded by the probability of generating sequences belonging to $\mathcal{X}^n - \mathcal{C}^n(R)$
(see Fig.~\ref{fig:CRn_shazou}). 
Hence we have the following:
\begin{align*} 
  &p_{\mathrm{e}}(\phin,\psin|p_X^n) \leq p_X^n(\mathcal{X}^n - \mathcal{C}^n(R)) \\
 &= \sum_{\bm{x} \in \mathcal{X}^n 
 - \mathcal{C}^n(R)} p_X^n(\bm{x})
   = \sum_{\substack{P \in \mathcal{P}_n({\cal X}):
       \\ 
       R  \leq H(P)}
        } p_X^n(T^n(P))
  \\
 &\MLeq{a} 
   \sum_{\substack{P \in \mathcal{P}_n({\cal X}): \\ 
                   R  \leq H(P) 
                   }} 
           \exp_2\{-nD(P||p_X)\}
  \\
\newcommand{\OmiTTRR}{
  &\leq \sum_{\substack{P \in \mathcal{P}_n({\cal X}): \\ 
        R  \leq H(P)}} 
 \exp_2\left\{
    -n \min_{\substack{P \in \mathcal{P}_n({\cal X}):\\ 
             R \leq H(P)} 
            } D(P||p_X)
        \right\}
\\}
&\leq |\mathcal{P}_n({\cal X})| 
 \exp_2\left\{
 -n \min_{\substack{P \in \mathcal{P}_n({\cal X}): \\ 
      R  \leq H(P)}} D(P||p_X)
        \right\}
\\
  &\MLeq{b} |\mathcal{P}_n({\cal X})|2^{-n E(R|p_X)}
  \MLeq{c} (n+1)^{|\mathcal{X}|} 
  2^{-nE(R|p_X)}.
\end{align*}
Step (a) follows from Lemma~\ref{lem:type}~part~c).
Step (b) follows from 
$\mathcal{P}_n(\mathcal{X}) \subseteq 
\mathcal{P}(\mathcal{X})$ and the definition of  
$E(R|p_X)$.
Step (c) follows from Lemma~\ref{lem:type}~part~a).
\hfill\IEEEQED

\subsection{Proof of Proposition \ref{pos:upb_Max-MI}}\label{prf:pos_upb_Max-MI}
In this subsection, we prove Proposition \ref{pos:upb_Max-MI} by analyzing the security of the encryption scheme proposed in Section~\ref{prf:thm_Direct}. \par

\noindent \underline{\textit{Upper Bound of $\Delta^{(n)}_{\mathrm{MI}}
(\Phin|p_X^n,p_K^n)$:}} \/ 
Let $U^m$ be a random variable 
with the uniform distribution 
$p_{U^m}$ over ${\cal X}^m$. 
On an upper bound of 
$\Delta^{(n)}_{\mathrm{MI}}(\Phin\mid p_X^n,p_K^n)$, 
we have the following: 
\begin{lemma} \label{lem:UnivCMaxMI}
  For the proposed 
  construction of $\Phin$, we have 
\begin{align}
\Delta^{(n)}_{\mathrm{MI}}(\Phin|p_X^n,p_K^n)
 &\leq m \log|\mathcal{X}|-H(\widetilde{K}^m)
\notag\\
&=D(p_{\tilde{K}^m}||p_{U^m}).
\label{lem:upbMI} 
\end{align}
\end{lemma}

\textit{Proof:}~
For the proposed construction of $\Phin$, 
we have ${C}^m = \widetilde{K}^m \oplus \phin({\bm{X}})$.
Then we have the following chain of inequalities:
\begin{align*}
&I({C}^m;{\bm{X}}) 
= H({C}^m)
- H(\widetilde{K}^m \oplus \phin({\bm{X}}) 
  |{\bm{X}}) 
\\
& \MEq{a}H({C}^m)- H(\widetilde{K}^m) 
  \leq m \log|\mathcal{X}| - H(\widetilde{K}^m).
\end{align*}
Step (a) follows from ${\bm{X}} \perp \bm{K}$. 
\hfill\IEEEQED

In the subsequent discussion, we evaluate the right members of (\ref{lem:upbMI}) in Lemma~\ref{lem:UnivCMaxMI}. 
Fix any $P\in \mathcal{P}_n({\cal X})$.
For each $(\bm{k},\widetilde{k}^m)\in T^n(P)
\times {\cal X}^m$, we define the following quantity:
\begin{align*}
 \Omega_{\bm{k};\varphi^{(n)}}
 (\widetilde{k}^m)
 \coloneq 
 \left\{
 \begin{array}{l}
  1, \mbox{ if } \varphi^{(n)}({\bm{k}})
 =\widetilde{k}^m,\\
 0,\mbox{ otherwise }.  
 \end{array}
\right.
\end{align*}
Furthermore, set 
\begin{align*}
&\Omega_{P;\varphi^{(n)}}(\widetilde{k}^m)
 \coloneq 
 \frac{1}{|T^n(P)|}
\sum_{\bm{k}\in T^n(P)}
 \Omega_{\bm{k};\varphi^{(n)}}
  (\widetilde{k}^m),\\
 & \Omega_{P ;\varphi^{(n)}}
 \coloneq
 \left\{
\Omega_{P;\varphi^{(n)}}
(\widetilde{k}^{m})
\right\}_{
\widetilde{k}^{m}
\in {\cal X}^{m}}. 
\end{align*}
The quantity $\Omega_{P ;\varphi^{(n)}}$
becomes a probability distribution on 
${\cal X}^{m}$.
By the definition of 
$\Omega_{P;\varphi^{(n)}}$, $P$ 
$\in {\cal P}_n({\cal X})$,
we have that for each 
$\widetilde{k}^{m}\in {\cal X}^{m}$, 
\begin{align}
& p_{\widetilde{K}^{m}}
(\widetilde{k}^{m}) 
= \sum_{P \in {\cal P}_n({\cal X}) }
\Omega_{P;\varphi^{(n)}}
(\widetilde{k}^{m}) p^n_{K}(T^n(P)).
\label{eqn:sssR}
\end{align}
On an upper bound of $D(p_{\tilde{K}^m}||p_{U^m})$,
we have the following lemma:
\begin{lemma}\label{lem:UnivCMaxMIb}
\begin{align*}
&D(p_{\tilde{K}^m}||p_{U^m}) 
\leq  
\sum_{P \in {\cal P}_n({\cal X}) }  
p^n_{K}(T^n(P)) D(\Omega_{P;\varphi^{(n)}}||
p_{U^m}).
\end{align*}
\end{lemma}

{\it Proof:} We have the following chain of inequalities: 
\begin{align*}
& D(p_{\tilde{K}^m}||p_{U^m}) 
=m \log |{\cal X}|-H(\tilde{K}^m)
\\
&\MLeq{a} m \log |{\cal X}|-
\sum_{P \in {\cal P}_n({\cal X}) }  
p^n_{K}(T^n(P)) H(\Omega_{P;\varphi^{(n)}})
\\
&=\sum_{P \in {\cal P}_n({\cal X}) }  
p^n_{K}(T^n(P)) D(\Omega_{P;\varphi^{(n)}}||
p_{U^m}).
\end{align*}
Step (a) follows from (\ref{eqn:sssR}) and 
the concave property of the entropy function 
with respect to the probability distribution.
\hfill\IEEEQED


In the subsequent discussion, we evaluate 
$D(\Omega_{P;\varphi^{(n)}}||$ $p_{U^m}), 
P\in {\cal P}_n({\cal X})$ 
in Lemma~\ref{lem:UnivCMaxMIb} 
based on the arguments 
in \cite{OohamaSec1}. 

\noindent \underline{\textit{Random Construction 
of Affine Encoders:}} \/ We first choose $m$ 
such that it satisfies (\ref{eqn:ChoosEm}).
By the definition of $\varphi^{(n)}$, we have that 
for $\bm{k} \in \mathcal{X}^n$,
$ \varphi^{(n)}(\bm{k}) = \bm{k} A \oplus b^m,$ 
where $b^m$ is a vector with $m$ columns.
Entries of $A$ and $b^m$ are from the field of $\mathcal{X}$. 
These entries are selected at random, 
independently of each other, and with a uniform distribution.

\newcommand{\commentoutB}{
  Randomly constructed linear encoder $\widetilde{\varphi}^{(n)}$ and affine encoder $\varphi^{(n)}$ have three properties shown in the following lemma.
\begin{lemma}
  $\quad$
  \begin{itemize} \label{lem:affine}
    \item[\textnormal{a)}] For any $x^n, v^n \in \mathcal{X}^n$ with $x^n \neq v^n$, we have
          \begin{align}
            &\Pr [\widetilde{\varphi}^{(n)}(x^n)= \widetilde{\varphi}^{(n)}(v^n)]\notag\\
            &= \Pr [(x^n \ominus v^n)A = 0^m] = |\mathcal{X}|^{-m}.
          \end{align}
    \item[\textnormal{b)}]  For any $s^n \in \mathcal{X}^n$ and for any $\widetilde{s}^m \in \mathcal{X}^m$, we have
          \begin{align}
            &\Pr [\varphi^{(n)}(s^n) = \widetilde{s}^m] \notag\\
            &= \Pr [s^n A \oplus b^m = \widetilde{s}^m] = |\mathcal{X}|^{-m}.
          \end{align}
    \item[\textnormal{c)}] For any $s^n, t^n$ with $s^n \neq t^n$ and for any $\widetilde{s}^m \in \mathcal{X}^m$, we have
          \begin{align}
            &\Pr [\varphi^{(n)}(s^n)= \varphi^{(n)}(t^n) = \widetilde{s}^m]\notag \\
            &= \Pr [s^n A \oplus b^m = t^n A \oplus b^m = \widetilde{s}^m] = |\mathcal{X}|^{-2m}.
          \end{align}
  \end{itemize}
\end{lemma}

Proof of Lemma \ref{lem:affine} 
is given in Santoso and Oohama \cite{OohamaSec1}.
}

\noindent 
\underline{\textit{Estimation 
of Approximation Error:}}
\/ For $P\in {\cal P}_n({\cal X})$, define the following function.
\begin{align*}
  \Theta_n(P) \coloneq 
  \log \left[1+ (|{\cal X}|^m-1)|T^n(P)|^{-1}\right].
\end{align*}
We let expectations based on the randomness 
of the encoder functions be denoted 
by ${\bf E}[\cdot]$.
Then, we have the following:
\begin{lemma} \label{lem:upbED}
For each $P\in {\cal P}_n({\cal X})$,  
\begin{align*}
& {\rm {\bf E}}\left[ 
D(\Omega_{P; \varphi^{(n)}}||p_{U^m})\right]
\leq \Theta_n(P).
\end{align*}
\end{lemma}

Proof of Lemma \ref{lem:upbED} is given 
in Appendix \ref{prf:upbED}. 
Proof of some extended version of this lemma 
is also given in Santoso and Oohama \cite{OohamaSec1}.
From this lemma we have the following corollary.
\begin{corollary}\label{cor:MI_upb}
$\exists\{\varphi^{(n)}\}_{n=1}^{\infty}$ with  $\varphi^{(n)}:\mathcal{X}^n \to \mathcal{X}^m$ 
such that $\forall P\in {\cal P}_n({\cal X})$,
 \begin{align}
 D(\Omega_{P; \varphi^{(n)}}||p_{U^m}) 
    \leq |{\cal P}_n({\cal X}) | \Theta_n(P).
 \label{eqn:DPUpBa}   
 \end{align}
\end{corollary}

{\it Proof:} We set 
$$
\zeta_n( \varphi^{(n)})
\defeq \sum_{P \in {\cal P}_n({\cal X}) }  
(\Theta_n(P))^{-1}
D(\Omega_{P; \varphi^{(n)}}||p_{U^m}).
$$
Note that $\zeta_n( \varphi^{(n)})$ is 
a random variable based on the random 
choice of $\varphi^{(n)}$. Computing 
${\rm E}\left[\zeta_n(\varphi^{(n)})\right]$, 
we have
\begin{align}
&{\bf E}\left[\zeta_n(\varphi^{(n)})\right]
={\bf E}\left[\sum_{P \in {\cal P}_n({\cal X}) }  
(\Theta_n(P))^{-1}
D(\Omega_{P;\varphi^{(n)}}||p_{U^m})
 \right]
\notag\\
& =\sum_{P \in {\cal P}_n({\cal X}) } 
(\Theta_n(P))^{-1}{\bf E}\left[
D(\Omega_{P;\varphi^{(n)}}||p_{U^m})
 \right]\leq |{\cal P}_n({\cal X})|.
\label{eqn:UpbZeta} 
\end{align}
From (\ref{eqn:UpbZeta}), we can see that 
$\exists \varphi^{(n)}:{\cal X}^n
 \to {\cal X}^m$ such that 
$$
\zeta_n(\varphi^{(n)})=
\sum_{P \in {\cal P}_n({\cal X}) }  
 (\Theta_n(P))^{-1}
D(\Omega_{P;\varphi^{(n)}}||p_{U^m})
 \leq |{\cal P}_n({\cal X})|,
$$
from which we have the 
bound (\ref{eqn:DPUpBa}) 
in Corollary \ref{cor:MI_upb}.
\hfill\IEEEQED

Combining Lemmas  \ref{lem:UnivCMaxMI} and \ref{lem:UnivCMaxMIb} and 
Corollary \ref{cor:MI_upb}, we have the following:
 \begin{lemma}\label{lem:Phi_upb}
$\exists\{\varphi^{(n)}\}_{n=1}^{\infty}$ with  $\varphi^{(n)}:\mathcal{X}^n \to \mathcal{X}^m$ 
such that 
$$
 \Delta^{(n)}_{\mathrm{MI}}(\Phin|p_X^n,p_K^n)
\leq |{\cal P}_n({\cal X})|
\sum_{P \in {\cal P}_n({\cal X}) }  
p^n_{K}(T^n(P))  \Theta_n(P).
$$
 \end{lemma}

For $\Theta_n(P), P\in {\cal P}_n({\cal X})$, 
we have the following:
\begin{lemma}\label{lem:Theta_upb}
For each $P\in {\cal P}_n({\cal X})$, we have  
  \begin{align*}
    \Theta_n(P) \leq 
    \left(R_n+\frac{1}{2}\right)
    (n+1)^{|{\cal X}|}2^{-n[H(P)-R_n]^{+}}.
  \end{align*}
\end{lemma}

Proof of Lemma \ref{lem:Theta_upb} is given 
in Appendix \ref{prf:lemThetaupb}.


\textit{Proof of Proposition 
\ref{pos:upb_Max-MI}:}
For simplicity of notation we set 
$\widetilde{R}_n\coloneq R_n+(1/2).$
Then we have the following: 
\begin{align*}
&\Delta^{(n)}_{\mathrm{MI}}(\Phin|p_X^n,p_K^n) 
\MLeq{a} |{\cal P}_n({\cal X})| 
\sum_{P \in {\cal P}_n({\cal X})}
\Theta_n(P) 2^{-nD(P||p_K)}
\\
&\MLeq{b}\widetilde{R}_n 
(n+1)^{|{\cal X}|} |{\cal P}_n({\cal X})|
\sum_{P \in {\cal P}_n({\cal X})}   
 2^{-n\left\{[H(P)-R_n]^{+}+
 D(P||p_K)\right\}} 
\\ 
 &\leq \widetilde{R}_n
 (n+1)^{|{\cal X}|}
  |{\cal P}_n({\cal X})|^2
 \\
&\quad \times  
\exp_2\left[-n \min_{P \in {\cal P}_n({\cal X})}  
\left\{[H(P)-R_n]^{+}+D(P||p_K)\right\}\right]
\\ 
&\MLeq{c} \widetilde{R}_n
(n+1)^{|{\cal X}|}|{\cal P}_n({\cal X})|^2
2^{-n[F(R|p_K)-\gamma_n]}
\\
&\MLeq{d}\widetilde{R}_n
(n+1)^{3|{\cal X}|} 
2^{-n[F(R|p_K)-\gamma_n]}.
\end{align*}
Step (a) follows from 
Lemma \ref{lem:type} part c) and Lemma \ref{lem:Phi_upb}. 
Step (b) follows from Lemma \ref{lem:Theta_upb}.
Step (c) follows from 
${\cal P}_n({\cal X})\subseteq {\cal P}({\cal X})$,
$[H(P)-R_n]^{+}\geq [H(P)-R]^{+}-\gamma_n $
and the definition of $F(R|p_K)$.
Step (d) follows from Lemma \ref{lem:type} part a).
\hfill \IEEEQED

\section{Proof of Proposition \ref{pos:Oohama25_1}}\label{prf:pos_Oohama25_1}
\newcommand{\commentoutC}{
\subsection{Proof of Proposition \ref{pos:HX_HK_com}}\label{prf:HX_HK_com}
In this section, we prove Proposition \ref{pos:HX_HK_com}, which is required for the proof of the converse theorem.
We present below the lemmas necessary for the proof of Proposition \ref{pos:HX_HK_com}.
  \begin{lemma} \label{HXtildes-l}
    For any $\alpha > 1$, we have
    \begin{align*}
      \frac{1}{n} H(\widetilde{\bm{X}}) \geq \min_{p_X} \left[ H(\widetilde{X}) + \alpha D(p_{\widetilde{X}}||p_X) \right] - \frac{\alpha}{n} \log \frac{1}{1 - \varepsilon}.
    \end{align*}
  \end{lemma}
  Proof of Lemma \ref{HXtildes-l} is given in Tyagi and Watanabe~\cite{Change_of_measure}.
  \begin{lemma} \label{FalphapX}
    For any $\alpha > 1$, we have the following:
    \begin{align*}
      &\min_{p_X} \left[ H(\widetilde{X}) + \alpha D(p_{\widetilde{X}}||p_X) \right] - \frac{\alpha}{n} \log \frac{1}{1 - \varepsilon} \\
      &\geq H(X) - 3\biggl[\frac{1}{2}\left(\log \frac{1}{p_{\min}}\right)^2\log|\mathcal{X}|\biggr]^{\frac{1}{3}}\biggl[\frac{1}{n}\log \frac{1}{1 - \varepsilon}\biggr]^{\frac{1}{3}}
    \end{align*}
  \end{lemma}
  Proof of Lemma \ref{FalphapX} is given in \ref{prf:FalphapX}.
  With these preparations, we prove Proposition \ref{pos:HX_HK_com}.

\textit{Proof of Proposition \ref{pos:HX_HK_com}:}
Assume that $R$ is $(\varepsilon,\delta)$-admissible.
Then, there exists $\{(\Phin,\Psin)\}^{\infty}_{n=1}$ such that
for all $\gamma > 0$, there exists $n_0 = n_0(\gamma) \in \mathbb{N}$ and for all $n \geq n_0$, we have the following:
\begin{align}
  &\frac{1}{n} \log |{\mathcal{C}}^{(n)}| \leq R+ \gamma,\label{gyaku1}                   \\
  &p_{\textrm{e}}(\phin,\psin|p_X^n) \leq \varepsilon, \label{gyaku2}                   \\
  &\Delta^{(n)}_{\textrm{MI}}(\Phin,\Psin) \leq \delta. \label{gyaku3}
\end{align}
We define the following function:
\begin{align*}
  \xi(\bm{X}) = \left\{\begin{array}{cl}
    0, &\mbox{if} ~\bm{X} \in \Dn,
    \vspace{0.2cm}                                \\
    1, &\mbox{otherwise}.
  \end{array}
  \right.
\end{align*}
Then, we have the following inequality:
\begin{align*}\label{upb:MI}
  &I(\bm{X};C^m) = I(\bm{X},\xi(\bm{X});C^m) = I(C^m;\xi(\bm{X}),\bm{X})\\
  &= I(C^m;\xi(\bm{X})) + I(C^m;\bm{X}|\xi(\bm{X})) \geq I(C^m, \bm{X}|\xi(\bm{X}))\\
  &= \Pr\{\bm{X} \in \Dn\} I(C^m;\bm{X}|\bm{X} \in \Dn) \\
  &\quad + \Pr\{\bm{X} \notin \mathcal{D}\} I(C^m;\bm{X}|\bm{X} \notin \Dn)\\
  &\geq \Pr\{\bm{X} \in \Dn\} I(C^m;\bm{X}|\bm{X} \in \Dn).
  \stepcounter{equation}\tag{\theequation}
\end{align*}
Here, based on the idea that it suffices to ensure security only for sequences that are correctly decoded, 
we define a security criterion by $\Delta_{\textrm{MI,Dec}}^{(n)} \coloneqq I(C^m;\bm{X}|\bm{X} \in \Dn)$.
In order to derive a lower bound of $\Delta_{\textrm{MI,Dec}}^{(n)}$, we define the following. (change of measure)
\begin{align}\label{ptildex}
  p_{\widetilde{\bm{X}}}(\bm{x})
  \coloneqq\left\{\begin{array}{cl}
    \frac{p_{\bm{X}}(\bm{x})}{\Pr\{\bm{X} \in \Dn\}},
     &\mbox{ if } \bm{X} \in \Dn
    \vspace{0.2cm}                                \\
    0,
     &\mbox{ otherwise }
  \end{array}
  \right.
\end{align}
In this case, we have the following chain of inequalities:
\begin{align*}\label{upb:MIDec}
  &I(C^m;\bm{X}|\bm{X} \in \Dn)\\
  & = H(\bm{X}|\bm{X} \in \Dn) - H(\bm{X}|C^m, \bm{X} \in \Dn)\\
  &= H(\widetilde{\bm{X}}) - H(\bm{X},C^m|C^m, \bm{X} \in \Dn)\\
  &= H(\widetilde{\bm{X}}) - H(\bm{X}, \PhiXnK |C^m, \bm{X} \in \Dn)\\
  &\MGeq{a} H(\widetilde{\bm{X}}) - H(\bm{X}, \bm{K}|C^m, \bm{X} \in \Dn)\\
  &= H(\widetilde{\bm{X}}) - H(\bm{K}|C^m, \bm{X} \in \Dn) \\
  &\quad - H(\bm{X}|\bm{K}, C^m, \bm{X} \in \Dn)\\
  &= H(\widetilde{\bm{X}}) - H(\bm{K}|C^m, \bm{X} \in \Dn)\\
  &\geq H(\widetilde{\bm{X}}) - H(\bm{K}|C^m)\\
  &\geq H(\widetilde{\bm{X}}) - H(\bm{K}) = H(\widetilde{\bm{X}}) - nH(K).
  \stepcounter{equation}\tag{\theequation}
\end{align*}
Step (a) follows from the data processing inequality.
From (\ref{upb:MIDec}), we have the following:
\begin{align*}\label{HK_Htilde}
  &\frac{I(\bm{X};C^m)}{\Pr\{\bm{X} \in \Dn\}} \MGeq{a} I(C^m;\bm{X}|\bm{X} \in \Dn)\\
  &\MGeq{b} H(\widetilde{\bm{X}}) - nH(K)\\
  &\MRar{c} \frac{\delta}{\Pr\{\bm{X} \in \Dn\}} \geq H(\widetilde{\bm{X}}) - nH(K)\\
  &\MRar{d}{\Rightarrow} \frac{\delta}{1-\varepsilon} \geq H(\widetilde{\bm{X}}) - nH(K)\\
  &\Rightarrow H(K) \geq \frac{1}{n} H(\widetilde{\bm{X}}) - \frac{\delta}{n(1-\varepsilon)}.
  \stepcounter{equation}\tag{\theequation}
\end{align*}
Step (a) follows from (\ref{upb:MI}).
Step (b) follows from (\ref{upb:MIDec}).
Step (c) follows from (\ref{gyaku3}).
Step (d) follows from (\ref{gyaku2}).
From (\ref{HK_Htilde}) and Lemmas \ref{HXtildes-l}, \ref{FalphapX}, 
we have Proposition \ref{pos:HX_HK_com}.
\hfill\IEEEQED
}

In this section, we prove Proposition~\ref{pos:Oohama25_1} used in the proof of the converse coding theorem.

Let $(\widetilde{C}^{(n)},\widetilde{\bm{X}}) \in \mathcal{C}^{(n)} \times \mathcal{X}^n$ be a random vector, whose joint distribution denoted by $q_{\widetilde{C}^{(n)}\widetilde{\bm{X}}}$ has the following form:
\begin{align*}
  q_{\widetilde{C}^{(n)}\widetilde{\bm{X}}}(c,\bm{x})
   = &\Pr\{(C^{(n)},\bm{X}) = (c,\bm{x})|\bm{X} \in \widetilde{\mathcal{B}}^{(n)}_{\gamma}\},\\
  &~(c,\bm{x}) \in \mathcal{C}^{(n)} \times \mathcal{X}^n.
\end{align*}
The following conditional distributions is important for later arguments.
\begin{align*}
  q_{\widetilde{C}^{(n)}|\bm{X}}(c|\bm{x}) = \Pr\{C^{(n)}=c|\bm{X} = \bm{x},\bm{X} \in \widetilde{\mathcal{B}}^{(n)}_{\gamma}\}.
\end{align*}
Furthermore, we define the probability $Q$ of the event $\bm{X} \in \widetilde{\mathcal{B}}^{(n)}_{\gamma}$ as
\begin{align*}
  Q \coloneqq p_X^n\left(\widetilde{\mathcal{B}}^{(n)}_{\gamma}\right) = \Pr\left\{\bm{X} \in \widetilde{\mathcal{B}}^{(n)}_{\gamma}\right\}
  = \sum_{\bm{x} \in \widetilde{\mathcal{B}}^{(n)}_{\gamma}} p_{\bm{X}}(\bm{x}).
\end{align*}
Then, for any $c \in \mathcal{C}^{(n)}$, we have
\begin{align*} \label{eq:qtildeC}
  q_{\widetilde{C}^{(n)}}(c)
   = \frac{1}{Q}\sum_{\bm{x} \in \widetilde{\mathcal{B}}^{(n)}_{\gamma}}p_{C^{(n)}|\bm{X}}(c|\bm{x})p_{\bm{X}}(\bm{x}).
  \stepcounter{equation}\tag{\theequation}
\end{align*}
\begin{lemma}\label{lem:upb_qtildeC}
  For any $\widetilde{C}^{(n)} \in \mathcal{C}^{(n)}$, we have
  \begin{align}\label{upb:qtildeC}
    q_{\widetilde{C}^{(n)}}(c) \leq Q^{-1}2^{-n[H(X)- \gamma]}.
  \end{align}
\end{lemma}

\textit{Proof:}
By the definition $\widetilde{\mathcal{B}}^{(n)}_{\gamma} = \widetilde{\mathcal{A}}^{(n)}_{\gamma} \cap \Dn$, we have
  \begin{align*}
    p_{\bm{X}}(\bm{x}) \leq 2^{-n[H(X)-
    \gamma]} \quad \mathrm{for}~ \bm{x} 
    \in \widetilde{\mathcal{B}}^{(n)}_{\gamma}.
  \end{align*}
Then, we have the following inequalities:
  \begin{align*}
    &q_{\widetilde{C}^{(n)}}(c)Q \MEq{a} \sum_{\bm{x} \in \widetilde{\mathcal{B}}^{(n)}_{\gamma}}
    p_{C^{(n)}|\bm{X}}(c|\bm{x})p_{\bm{X}}(\bm{x})\\
    &\MLeq{b} 2^{-n[H(X)- \gamma]}\sum_{\bm{x} \in \widetilde{\mathcal{B}}^{(n)}_{\gamma}}p_{C^{(n)}|\bm{X}}(c|\bm{x})
     \MLeq{c} 2^{-n[H(X)- \gamma]}.
  \end{align*}
Step (a) follows from (\ref{eq:qtildeC}).  
Step (b) follows from (\ref{upb:qtildeC}).  
Step (c) follows from Lemma \ref{lem:BirkoffvN}.  
\hfill\IEEEQED

We next derive a lower bound of  
 $H(C^{(n)}|\bm{X} \in \widetilde{\mathcal{B}}^{(n)}_{\gamma}) 
  = H(\widetilde{C}^{(n)}).$
This lower bound is given by the following lemma:
\begin{lemma}\label{lem:lwb_H(Ctilde)}
  \begin{align}\label{lwb:H(Ctilde)}
    H(\widetilde{C}^{(n)}) \geq nH(X) + \log Q.
  \end{align}
\end{lemma}

\textit{Proof:}
We have the following chain of inequalities:
  \begin{align*}\label{prf:lwb_H(Ctilde)}
    &H(\widetilde{C}^{(n)}) = \sum_{c \in \mathcal{C}^{(n)}} q_{\widetilde{C}^{(n)}}(c) \log \frac{1}{q_{\widetilde{C}^{(n)}}(c)}\\
                           &\MGeq{a} \sum_{c \in \mathcal{C}^{(n)}} q_{\widetilde{C}^{(n)}}(c) \log \left[ Q~ 2^{n[H(X)-\gamma]} \right]
\\
& = n\left[H(X) - \gamma \right] + \log Q.
\stepcounter{equation}\tag{\theequation}
\end{align*}
Since $\gamma$ can be chosen arbitrarily small in (\ref{prf:lwb_H(Ctilde)}), 
we have the bound (\ref{lwb:H(Ctilde)}) in Lemma \ref{lem:lwb_H(Ctilde)}
\hfill\IEEEQED

We next evaluate 
$ H(C^{(n)}|\bm{X},\bm{X} \in \widetilde{\mathcal{B}}^{(n)}_{\gamma}) = H(\widetilde{C}^{(n)}|\widetilde{\bm{X}}).$
On an upper bound of this quantity we have 
the following: 
\begin{lemma}\label{lem:lwb_H(Ctilde|X)}
  \begin{align}\label{upb:H(Ctilde|X)}
    H(C^{(n)}|\bm{X}, \bm{X} \in \widetilde{\mathcal{B}}^{(n)}_{\gamma}) \leq nH(K).
  \end{align}
\end{lemma}

\textit{Proof:}
We have the following chain of inequalities:
  \begin{align*}
    &H(C^{(n)}|\bm{X}, \bm{X} \in \widetilde{\mathcal{B}}^{(n)}_{\gamma}) 
    =H(\PhiXnK|\bm{X}, \bm{X} \in \widetilde{\mathcal{B}}^{(n)}_{\gamma}) \\
    &\MLeq{a} H(\bm{K}|\bm{X}, \bm{X} \in \widetilde{\mathcal{B}}^{(n)}_{\gamma})
    \MEq{b} nH(K).
  \end{align*}
Step (a) follows from the data processing inequality.  
Step (b) follows since $\bm{K} \perp \bm{X}$. 
\hfill\IEEEQED

The following lemma gives a relationship between security, reliability and mutual information.
\begin{lemma}\label{lem:upb_cndI}
  We have
  \begin{align}\label{upb:lem_condI}
    Q^{-1}\delta \geq I(\bm{X};C^{(n)}|\bm{X} \in \widetilde{\mathcal{B}}^{(n)}_{\gamma}).
  \end{align}
\end{lemma}

\textit{Proof:} Let  $\xi(\bm{X})\in \{0,1\}$ 
be a binary random variable such that
if $\bm{X} \in \widetilde{\mathcal{B}}^{(n)}_{\gamma}$,
$\xi(\bm{X})=1$, otherwise, $\xi(\bm{X})=0$.   
\newcommand{\Sdfzz}{
Define
\begin{align*}
  \xi(\bm{X}) \coloneqq 
  \left\{\begin{array}{cl}
    &\hspace*{-11pt}1,\mbox{ if } \bm{X} \in \widetilde{\mathcal{B}}^{(n)}_{\gamma},\\
    &\hspace*{-11pt}0,\mbox{ otherwise }. 
  \end{array}
\right.
\end{align*}
}
Let 
$\overline{Q}:= 1-Q
=\Pr\left\{ 
\bm{X} \notin \widetilde{\mathcal{B}}^{(n)}_{\gamma}
\right\}$. On lower bounds of $I(\bm{X};C^{(n)})$, 
we have the following chain of inequalities:
\begin{align*}\label{upb:MIwithQ}
  &I(\bm{X};C^{(n)}) = I(\bm{X}, \xi(\bm{X});C^{(n)})
   \geq I(\bm{X};C^{(n)}|\xi(\bm{X}))\\
  &= Q I(\bm{X};C^{(n)}|\bm{X} \in \widetilde{\mathcal{B}}^{(n)}_{\gamma}) 
 + \overline{Q}I(\bm{X};C^{(n)}|\bm{X} \notin \widetilde{\mathcal{B}}^{(n)}_{\gamma})\\
  &\geq Q I(\bm{X};C^{(n)}|\bm{X} \in \widetilde{\mathcal{B}}^{(n)}_{\gamma}).
  \stepcounter{equation}\tag{\theequation}
\end{align*}
From (\ref{upb:MIwithQ}), we have
\begin{align*}
  Q^{-1} I(\bm{X};C^{(n)}) \geq  I(\bm{X};C^{(n)}|\bm{X} \in \widetilde{\mathcal{B}}^{(n)}_{\gamma}).
\end{align*}
Since $\delta \geq I(\bm{X};C^{(n)})$, we get the bound (\ref{upb:lem_condI}) in Lemma \ref{lem:upb_cndI}.
\hfill\IEEEQED


 \textit{Proof of Proposition~\ref{pos:Oohama25_1}:}
For a lower bound of $I(\bm{X};C^{(n)}$ $|\bm{X} \in \widetilde{\mathcal{B}}^{(n)}_{\gamma})$, we have the following chain of inequalities:
\begin{align*}\label{lwb:pos_condI}
  &I(\bm{X};C^{(n)}|\bm{X} \in \widetilde{\mathcal{B}}^{(n)}_{\gamma})
  = H(C^{(n)}|\bm{X} \in \widetilde{\mathcal{B}}^{(n)}_{\gamma}) 
\\  
&\quad - H(C^{(n)}|\bm{X}, \bm{X} \in \widetilde{\mathcal{B}}^{(n)}_{\gamma})
  \MGeq{a} H(C^{(n)}) -nH(K)\\
  &\MGeq{b} n[H(X)-\gamma] + \log Q -nH(K).
  \stepcounter{equation}\tag{\theequation}
\end{align*}
Step (a) follows from (\ref{upb:H(Ctilde|X)}).  
Step (b) follows from (\ref{lwb:H(Ctilde)}).
From (\ref{upb:lem_condI}) and (\ref{lwb:pos_condI}), we have
\begin{align}\label{ineq:H(X)_H(K)_Q}
  &Q^{-1} \delta \geq n[H(X)-\gamma] + \log Q -nH(K)\notag\\
  &\Leftrightarrow H(X) \leq H(K) + \gamma + \frac{1}{n}\left[\frac{\delta}{Q} + \log \frac{1}{Q}\right].
\end{align}
For (\ref{ineq:H(X)_H(K)_Q}), we have
\begin{align*}
  &H(X) \leq H(K) + \gamma + \frac{1}{n}\left[\frac{\delta}{Q} + \log \frac{1}{Q}\right]\\
  &\MLeq{a} H(K) + \gamma + \frac{1}{n}\left[\frac{\delta}{1-\widetilde{\nu}_n(\gamma, \varepsilon)} + \log \frac{1}{1-\widetilde{\nu}_n(\gamma, \varepsilon)}\right].
\end{align*}
Step (a) follows from that 
the bound (\ref{upb:Pr_B^c}) 
is equivalent to 
${Q} \geq 1-\widetilde{\nu}_n($$\gamma, \varepsilon)$. 
\hfill \IEEEQED

\section{Conclusions}

In this study, based on previous works, we proposed 
a framework for source encryption. Instead of 
the security measures used in the previous studies, 
we adopted mutual information as the security 
criterion. Under this criterion, we obtained 
the necessary and sufficient condition on 
$R$ and $(p_X,p_K)$ for simultaneously 
achieving communication efficiency and security. 
The obtained condition does not depend on the pair  $(\varepsilon,\delta)\in (0,1)$$\times(0,\delta_0]$ 
of reliable and security constants, 
implying that we have the 
strong converse coding theorem for the proposed
framework of source encryption.
We also establish a universal construction 
of $\left\{(\Phin,\Psin)\right\}_{n=1}^{\infty}$
attaining the condition in the sense that 
$\left\{(\Phin,\Psin)\right\}_{n=1}^{\infty}$   
does not depend on $(p_X,p_K) \in {\cal P}^2({\cal X})$.


\bibliographystyle{IEEEtran}
\bibliography{isita2026_FF_FV}

\begin{thebibliography}{10}
\providecommand{\url}[1]{#1}
\csname url@samestyle\endcsname
\providecommand{\newblock}{\relax}
\providecommand{\bibinfo}[2]{#2}
\providecommand{\BIBentrySTDinterwordspacing}{\spaceskip=0pt\relax}
\providecommand{\BIBentryALTinterwordstretchfactor}{4}
\providecommand{\BIBentryALTinterwordspacing}{\spaceskip=\fontdimen2\font plus
\BIBentryALTinterwordstretchfactor\fontdimen3\font minus \fontdimen4\font\relax}
\providecommand{\BIBforeignlanguage}[2]{{%
\expandafter\ifx\csname l@#1\endcsname\relax
\typeout{** WARNING: IEEEtran.bst: No hyphenation pattern has been}%
\typeout{** loaded for the language `#1'. Using the pattern for}%
\typeout{** the default language instead.}%
\else
\language=\csname l@#1\endcsname
\fi
#2}}
\providecommand{\BIBdecl}{\relax}
\BIBdecl

\bibitem{Shannon}
C.~E. Shannon, ``Communication theory of secrecy systems,'' \emph{The Bell System Technical Journal}, vol.~28, no.~4, pp. 656--715, 1949.

\bibitem{Yamamoto1}
H.~Yamamoto, ``Coding theorems for {S}hannon's cipher system with correlated source outputs, and common information,'' \emph{IEEE Transactions on Information Theory}, vol.~40, no.~1, pp. 85--95, 1994.

\bibitem{Yamamoto2}
------, ``Rate-distortion theory for the {S}hannon cipher system,'' \emph{IEEE Transactions on Information Theory}, vol.~43, no.~3, pp. 827--835, 1997.

\bibitem{Yamamoto3}
Y.~Hayashi and H.~Yamamoto, ``Coding theorems for the {S}hannon cipher system with a guessing wiretapper and correlated source outputs,'' \emph{IEEE Transactions on Information Theory}, vol.~54, no.~6, pp. 2808--2817, 2008.

\bibitem{DBLP:conf/isit/OohamaS22}
\BIBentryALTinterwordspacing
Y.~Oohama and B.~Santoso, ``A framework for {S}hannon ciphers under side-channel attacks: A strong converse and more,'' in \emph{{IEEE} International Symposium on Information Theory, {ISIT} 2022, Espoo, Finland, June 26 - July 1, 2022}.\hskip 1em plus 0.5em minus 0.4em\relax {IEEE}, 2022, pp. 862--867. [Online]. Available: \url{https://doi.org/10.1109/ISIT50566.2022.9834899}
\BIBentrySTDinterwordspacing

\bibitem{DBLP:conf/isit/OohamaS24}
\BIBentryALTinterwordspacing
------, ``Universal source encryption under side-channel attacks,'' in \emph{{IEEE} International Symposium on Information Theory, {ISIT} 2024, Athens, Greece, July 7-12, 2024}.\hskip 1em plus 0.5em minus 0.4em\relax {IEEE}, 2024, pp. 3344--3349. [Online]. Available: \url{https://doi.org/10.1109/ISIT57864.2024.10619496}
\BIBentrySTDinterwordspacing

\bibitem{DBLP:conf/itw/OohamaS21}
\BIBentryALTinterwordspacing
------, ``Strong converse for distributed source coding with encryption using correlated keys,'' in \emph{{IEEE} Information Theory Workshop, {ITW} 2021, Kanazawa, Japan, October 17-21, 2021}.\hskip 1em plus 0.5em minus 0.4em\relax {IEEE}, 2021, pp. 1--6. [Online]. Available: \url{https://doi.org/10.1109/ITW48936.2021.9611414}
\BIBentrySTDinterwordspacing

\bibitem{DBLP:conf/isita/OohamaS22}
\BIBentryALTinterwordspacing
------, ``A framework for distributed source coding with encryption: A new strong converse and more,'' in \emph{International Symposium on Information Theory and Its Applications, {ISITA} 2022, Tsukuba, Ibaraki, Japan, October 17-19, 2022}.\hskip 1em plus 0.5em minus 0.4em\relax {IEEE}, 2022, pp. 189--193. [Online]. Available: \url{https://ieeexplore.ieee.org/document/10683942}
\BIBentrySTDinterwordspacing

\bibitem{Oohama2025}
------, ``Strong converse for distributed source encryption under standard mutual information,'' in \emph{2025 IEEE Information Theory Workshop (ITW)}, 2025, pp. 632--637.

\bibitem{Han98InfSpec}
T.~S. Han, \emph{Information-Spectrum Methods in Information Theory}.\hskip 1em plus 0.5em minus 0.4em\relax Springer, 2003.

\bibitem{Iwamoto}
M.~Iwamoto and K.~Ohta, ``Security notions for information theoretically secure encryptions,'' in \emph{2011 IEEE International Symposium on Information Theory Proceedings}, 2011, pp. 1777--1781.

\bibitem{HanKobayashi}
T.~S. Han and K.~Kobayashi, \emph{Mathematics of Information and Coding}, ser. Translation of Mathematical Monographs, S.~Kobayashi and M.~Takesaki, Eds.\hskip 1em plus 0.5em minus 0.4em\relax American Mathematical Society, 2002, vol. 203.

\bibitem{OohamaSec1}
\BIBentryALTinterwordspacing
B.~Santoso and Y.~Oohama, ``Information theoretic security for {S}hannon cipher system under side-channel attacks,'' \emph{Entropy}, vol.~21, no.~5, 2019. [Online]. Available: \url{https://www.mdpi.com/1099-4300/21/5/469}
\BIBentrySTDinterwordspacing

\end{thebibliography}


\appendix

\subsection{Proof of Property \ref{per:Dn}}\label{prf:Dn}
%

Under $\bm{x}, \bm{x}^{\prime} \in \Dn$ and 
$\bm{x} \neq \bm{x}^{\prime}$, we assume that
\begin{equation}
  \Phiknx = \Phi_{\bm{k}}^{(n)}(\bm{x}^{\prime}).
  \label{eqn:Assumm}
\end{equation}
Then, we have the following:
\begin{align}
  &\bm{x} \MEq{a} \psin \circ \phinx \MEq{b} \Psi_{\bm{k}}^{(n)} \circ \Phiknx\notag                                             \\
         &\MEq{c} \Psi_{\bm{k}}^{(n)} \circ \Phi_{\bm{k}}^{(n)}(\bm{x}^{\prime}) \MEq{d} \psin \circ\phin(\bm{x}^{\prime}) \MEq{e} \bm{x}^{\prime}.
  \label{eqn:SdCCvv}
\end{align}
Step (a) and (e) follows from definition of $\Dn$. 
Step (c) follows from (\ref{eqn:Assumm}).
Step (b) and (d) follow from the relationship between $(\phin,\psin)$ and $(\Phi_{\bm{k}}^{(n)},\Psi_{\bm{k}}^{(n)})$.
Equation~(\ref{eqn:SdCCvv}) contradicts 
$\bm{x} \neq \bm{x}^{\prime}$. Hence we must have $\Phiknx \neq \Phi_{\bm{k}}^{(n)}(\bm{x}^{\prime})$.
\hfill\IEEEQED

\subsection{Proof of Lemma \ref{lem:BirkoffvN}}\label{prf:Birkholf}

For $\bm{x} \in \mathcal{X}^n$, we define 
$\mathcal{A}_{\bm{x}}(c) \coloneqq \left\{\bm{k}:
 \Phin_{\bm{x}}({\bm{k}}) = c \right\}.
$
Then we have following:
\begin{align*} \label{eqn:ddrtrqqqq}
 &p_{C^{(n)}|\bm{X}}(c|\bm{x}) 
 = \Pr\left\{\PhiKn(\bm{x}) = c 
  \Bigl|\bm{X} = \bm{x} \right\}
  \\
 &
 \MEq{a} \Pr\left\{\bm{K} 
 \in \mathcal{A}_{\bm{x}}(c)\Bigl|
 \bm{X}= \bm{x}  \right\} 
\MEq{b} \Pr\left\{\bm{K}
 \in \mathcal{A}_{\bm{x}}(c)\right\}.
  \stepcounter{equation}\tag{\theequation}
\end{align*}
Step (a) follows from $\bm{K} \in \mathcal{A}_{\bm{x}}(c) \Leftrightarrow \Phi_{\bm{K}}(\bm{x}) = c$.
Step (b) follows from $\bm{K} \perp \bm{X}$.
On the other hand, by Property~\ref{per:Dn}, 
\begin{align} \label{eqn:Serrq}
   &\mathcal{A}_x(c) \cap \mathcal{A}_{{\bm{x}}^{\prime}} (c)=\emptyset ,~\mbox{ for }\bm{x}  \neq {\bm{x}}^{\prime} \in \Dn.
\end{align}
From the above, we have the following chain of equalities.
\begin{align*}
  &\sum_{\substack{\bm{x} \in \Dn}} p_{C^{(n)}|\bm{X}}(c|\bm{x}) \MEq{a} \sum_{\substack{\bm{x} \in \Dn}} \Pr \left\{\bm{K} \in \mathcal{A}_{\bm{x}}(c) \right\}\\
  &\MEq{b} \Pr \left\{\bm{K} \in \bigcup_{\substack{\bm{x} \in \Dn }} \mathcal{A}_{\bm{x}}(c) \right\} \leq 1
\end{align*}
Steps (a) and (b), respectively, are from (\ref{eqn:ddrtrqqqq}) and 
(\ref{eqn:Serrq}).
\hfill\IEEEQED

\newcommand{\PrfLemUpdEd}{
\subsection{Proof of Lemma \ref{lem:upbED}\
}\label{prf:upbED}

In this appendix, we prove Lemma \ref{lem:upbED}. 
\newcommand{\commentoutaB}{
  }{  
  Randomly constructed 
  affine encoder $\varphi^{(n)}$ have two properties 
  shown in the following lemma.
\begin{lemma}\label{lem:good_set}
  $\quad$
  \begin{itemize} 
  \item[\textnormal{a)}]  For any $\bm{s} \in \mathcal{X}^n$ 
  and for any $\widetilde{s}^m \in \mathcal{X}^m$, we have
    \begin{align}
     &\Pr [\varphi^{(n)}(\bm{s}) = \widetilde{s}^m] 
     = \Pr [\bm{s} A \oplus b^m = \widetilde{s}^m] 
      = |\mathcal{X}|^{-m}.
          \end{align}  
  \item[\textnormal{b)}] For any $\bm{s}, \bm{t}$ 
  with $\bm{s} \neq \bm{t}$ and 
  for any $\widetilde{s}^m \in \mathcal{X}^m$, we have
  \begin{align}
      &\Pr [\varphi^{(n)}(\bm{s})= \varphi^{(n)}(\bm{t}) 
       = \widetilde{s}^m]\notag \\
      &= \Pr [\bm{s} A \oplus b^m = \bm{t} A \oplus b^m 
       = \widetilde{s}^m]=|\mathcal{X}|^{-2m}.
          \end{align}
  \end{itemize}
\end{lemma}

Proof of Lemma \ref{lem:good_set} 
is given in Santoso and Oohama \cite{OohamaSec1}.
}

Fix $P\in {\cal P}_n({\cal X}).$ Let $\bm{K}_P$ be 
a random variable uniformly distributed over 
$T^n(P)$. 
In~the following arguments, we use the following 
simplified notations: 
\begin{align*}
 \bm{k} , \bm{K}_P \in T^n(P) \subseteq {\cal X}^n 
 &\Longrightarrow  k , K \in {\cal T} \subseteq {\cal K}, 
\\
 \tilde{k}^m ,\tilde{K}^m \in {\cal X}^m 
 &\Longrightarrow l, L \in {\cal L}, 
\\
\varphi^{(n)}: {\cal X}^n \to {\cal X}^m &\Longrightarrow 
\varphi: {\cal K} \to {\cal L}, 
\\
\varphi^{(n)}(\bm{k})=\bm{k}A +b^m 
& \Longrightarrow  \varphi(k) =k A+b, 
\\
 U^m \in{\cal X}^m & \Longrightarrow  U \in {\cal L},
\end{align*}
We define
$$
\chi_{\varphi(k),l}=
\left\{
\begin{array}{l}
 1,\mbox{ if }\varphi(k)=l, 
\\
0,\mbox{ if }\varphi(k) \neq l.  
\end{array}
\right.
$$
Then, the~distribution
of the random variable $L= L_{\varphi}$ 
is 
\begin{align*}
&p_{L}(l)
=\sum_{k \in {\cal T}} p_{K} (k)
\chi_{\varphi(k),l} \mbox{ for }l \in {\cal L}.
\end{align*}
Define 
\begin{align*}
& \Upsilon_{ \varphi(k),l}
\defeq \chi_{\varphi(k),l}
\log \hugebl |{\cal L}|\hugel
\sum_{k^{\prime} \in {\cal T}} p_{K}(k^{\prime})
\chi_{\varphi(k^{\prime}),l}
\huger \hugebr.
\end{align*}
Then the divergence  between 
$p_{L}$ and $p_{U}$ is given by
\begin{align}
&\left. \left. 
 D\left(p_{L}\right|\right|p_{U} \right)
=\sum_{k\in {\cal T}} 
 \sum_{l\in {\cal L}}p_{K} (k) 
\Upsilon_{\varphi(k),l}.
\end{align}
The quantity 
$\Upsilon_{\varphi(k),l} $ has the following form:
\begin{align}
& \Upsilon_{\varphi(k),l}=\chi_{\varphi(k),l}
\log \Biggl\{ |{\cal L}|\Biggl(p_{K}(k)\chi_{\varphi(k),l}
\notag\\
& 
\left. \left. 
 + \sum_{k^{\prime} \in \{k\}^{\rm c}} 
p_{K}(k^{\prime}) 
\chi_{\varphi(k^{\prime}),l}
\right)\right\}.
\label{eqn:AzzxW}
\end{align}
The above form is useful for computing 
${\bf E}[ \Upsilon_{\varphi(k),l}]$. 

{\it Proof of Lemma \ref{lem:upbED}:}
Taking the expectation of both  sides of 
(\ref{eqn:AzzxW})
with respect to the random choice of the entry 
of the matrix $A$ and the vector $b$ representing 
the affine encoder $\varphi$, 
we have
\begin{align}
&\left. \left. {\bf E}\left[ D \left(p_{L}\right|\right|
p_{U}                      \right)\right]
=\sum_{k \in {\cal T}}
 \sum_{l \in {\cal L}}p_{K}(k)
{\bf E}\left[\Upsilon_{\varphi(k),l}\right].
\label{eqn:Zdxxp}
\end{align}
To compute the expectation 
${\bf E}\left [\Upsilon_{\varphi(k),l}\right]$, 
we introduce an expectation operator useful for the computation.
Let ${\bf E}_{\varphi(k)=l_k}[\cdot]$ 
be an expectation operator based on the conditional probability measures
${\rm Pr}(\cdot|\varphi(k)=l_k)$.
Using this expectation operator, the~quantity 
${\bf E} \left[\Upsilon_{\varphi(k),l}\right]$ 
can be written as
\begin{align}
&{\bf E}\left[\Upsilon_{\varphi(k),l}\right]
=\sum_{ l_k \in {\cal L}}{\rm Pr} \left(\varphi(k)=l_k\right)
{\bf E}_{\varphi(k)=l_k}\left[\Upsilon_{l_k,l}\right].
\label{eqn:SddXPP}
\end{align}
Note that 
\beq
\Upsilon_{l_k,l}
=\left\{
\ba{l}
1, \mbox{ if } l_k=l,\\
0, \mbox{ otherwise.}
\ea
\right.
\label{eqn:SdXl}
\eeq
From (\ref{eqn:SddXPP}) and (\ref{eqn:SdXl}), 
we have
\begin{align}
&{\bf E} \left[\Upsilon_{\varphi(k),l}\right]
={\rm Pr} \left(\varphi(k)=l\right)
{\bf E}_{\varphi(k)=l}
\left[\Upsilon_{l,l}\right]
\notag\\
&
=\frac{1}{|{\cal L}|}{\bf E}_{\varphi(k)=l}
\left[\Upsilon_{l,l}\right].
\label{eqn:ASdff}
\end{align}
Using (\ref{eqn:AzzxW}), the~expectation 
${\bf E}_{\varphi(k)=l}\left[\Upsilon_{l,l}\right]$ 
can be written as
\begin{align}
& {\bf E}_{\varphi(k)=l} \left[ \Upsilon_{l,l} \right] 
={\bf E}_{\varphi(k)=l} 
  \Biggl[
  \log \Biggl\{ |{\cal L}|\Biggl(p_{K}(k)
  \notag\\
& 
\left.\left.\left. 
\qquad 
+ 
\sum_{k^{\prime} \in \{k\}^{\rm c}} 
p_{K}(k^{\prime}) 
\chi_{\varphi(k^{\prime}),l}
\right)\right\}\right].
\label{eqn:AzzxWcc}
\end{align}
Applying Jensen's inequality to the right member of (\ref{eqn:AzzxWcc}), 
we obtain the following upper bound of 
${\bf E}_{\varphi(k)=l} \left[\Upsilon_{l,l}\right]$:
\begin{align}
& {\bf E}_{\varphi(k)=l} \left[\Upsilon_{l,l}\right]
\leq 
\log \Biggl\{ |{\cal L}|\Biggl(p_{K}(k)
\notag\\
& 
\left.\left.
\qquad 
+ \sum_{k^{\prime} \in \{k\}^{\rm c}} 
        p_{K}(k^{\prime})
 {\bf  E}_{\varphi(k)=l}\left[\chi_{\varphi(k^{\prime}),l}\right] 
\right)\right\}
\notag\\
&\MEq{a} 
\log \Biggl\{|{\cal L}|\Biggl(
        p_{K}(k) + \sum_{k^{\prime} \in \{k\}^{\rm c}} 
        p_{K}(k^{\prime})\frac{1}{|{\cal L}|}
\Biggr)\Biggr\}
\notag\\
&= \log \left\{1+ (|{\cal L}|-1)p_{K}(k)\right\}.
\label{eqn:AzzxWfd}
\end{align}

Step (a) follows from that by Lemma \ref{lem:good_set} 
parts a) and b), 
\begin{align*}
{\bf E}_{ \varphi(k)=l}\left[\chi_{\varphi(k^{\prime}),l}\right] 
 &={\rm Pr} (\varphi(k^{\prime})
 =l|\varphi(k)=l)=\frac{1}{|{\cal L}|}.
\end{align*}
From (\ref{eqn:Zdxxp}), (\ref{eqn:ASdff}), 
and~(\ref{eqn:AzzxWfd}), we have 
the following:
\begin{align*}
& \left. \left. {\bf E}\left[ D \left(p_{L}\right|\right|
p_{U}                      \right)\right]
\leq \sum_{k \in {\cal T}}
 \sum_{l \in {\cal L}}p_{K}(k)
\frac{1}{|{\cal L}|} 
\\
& \quad\times \log \left\{1+ (|{\cal L}|-1)p_{K}(k)\right\}
\MEq{a}
\log 
\left\{1+ (|{\cal L}|-1)|{\cal T}|^{-1}\right\}.
\end{align*}
Step (a) follows from that $K\in {\cal T}$ is 
the uniformly distributed random variable 
over the set ${\cal T}$.  
\hfill\IEEEQED
}

\PrfLemUpdEd

\subsection{Proof of Lemma~\ref{lem:Theta_upb}}\label{prf:lemThetaupb}

%
For each $P\in {\cal P}_n({\cal X})$, 
we have 
\begin{align}
& \Theta_n(P) \MLeq{a} 
   \log \left\{1+(2^{nR_n}-1) 
 |T^n(P)|^{-1} \right\}
\label{eqn:ThetaBdOne}\\
& \leq nR_n \leq \left(R_n+\frac{1}{2}\right) (n+1)^{|{\cal X}|}.
\label{eqn:ThetaBdTwo}
\end{align}
Step (a) follows from the definition of 
$\Theta_n(P)$ and 
$|{\cal X}|^m \leq 2^{nR_n}$. 
We also have the following chain of inequalities: 
\begin{align} 
& \Theta_n(P) \MLeq{a} 
  (2^{nR_n}-1) |T^n(P)|^{-1}\leq 2^{nR_n}|T^n(P)|^{-1} 
\notag\\  
&\MLeq{b}(n+1)^{|{\cal X}|-1}2^{-n[H(P)-R_n]}
\notag\\
&\MLeq{c}
\left(R_n+\frac{1}{2}\right) (n+1)^{|{\cal X}|}
2^{-n[H(P)-R_n]}.
\label{eqn:ThetaBdThr}
\end{align}
Step (a) follows from (\ref{eqn:ThetaBdOne}) and  
$\log(1+a) \leq a$. 
Step (b) follows from Lemma \ref{lem:type} part b).
Step (c) follows from 
$$
 (n+1)^{|{\cal X}|-1}\leq 
 (nR_n+1) (n+1)^{|{\cal X}|-1}
 \leq \left(R_n+\frac{1}{2}\right)(n+1)^{|{\cal X}|}.
$$
From (\ref{eqn:ThetaBdTwo}) and (\ref{eqn:ThetaBdThr}), 
we have the bound 
in Lemma \ref{lem:Theta_upb}.
\hfill\IEEEQED

\end{document}